\documentclass[abstract,10pt]{scrartcl}
\pdfoutput=1
\usepackage[utf8]{inputenc}
\usepackage[T1]{fontenc}
\usepackage[english]{babel}
\usepackage{amsmath,amssymb,amsthm}
\usepackage{hyphenat}
\usepackage[textsize=footnotesize]{todonotes}
\usepackage[numbers,sort&compress]{natbib}
\usepackage{microtype}
\usepackage{authblk}
\usepackage{newtxtext}
\usepackage{newtxmath}
\usepackage{tikz}
\usepackage{pgfplots}
\usepackage{subcaption}

\DeclareMathOperator*{\argmin}{arg\,min}
\DeclareMathOperator*{\indeg}{indeg}
\DeclareMathOperator*{\outdeg}{outdeg}

\usetikzlibrary{calc,decorations.pathmorphing}

\newcommand{\eG}[2]{\ensuremath{#1\langle #2\rangle}}
\newcommand{\OPT}{\text{OPT}}
\newcommand{\pWR}{\beta}
\newcommand{\pBal}{\gamma}
\newcommand{\pM}{\delta}
\usepackage{hyperref}
\usepackage{doi}
\usepackage[sort&compress,nameinlink,noabbrev,capitalize]{cleveref}
\usepackage{paralist}
\usepackage{tabularx}
\newcommand\setrow[1]{\gdef\rowmac{#1}#1\ignorespaces}
\newcommand\clearrow{\global\let\rowmac\relax}
\clearrow
\usepackage{booktabs}
\usepackage{fullpage}
\hypersetup{colorlinks,allcolors=blue}

\numberwithin{equation}{section}
\numberwithin{figure}{section}

\newcommand{\EE}{Eulerian extension}
\newcommand{\M}{\mathcal M}
\newcommand{\sol}{RPP tour}
\newcommand{\ba}{balanced}
\newcommand{\imba}{imbalanced}
\DeclareMathOperator{\poly}{poly}

\newcommand{\emi}{edge\hyp minimizing}

\newcommand{\cost}{\omega}
\newcommand{\N}{\mathbb{N}}
\newcommand{\apxfakt}{\alpha}
\newcommand{\Q}{\mathbb{Q}}

\newtheoremstyle{sans}
  {\topsep}   {\topsep}   {}  {0pt}       {\bfseries\sffamily} {.}         {5pt plus 1pt minus 1pt} {}          

\theoremstyle{sans}
\newtheorem{theorem}{Theorem}[section]
\newtheorem{observation}[theorem]{Observation}
\newtheorem{problem}[theorem]{Problem}
\newtheorem{lemma}[theorem]{Lemma}
\newtheorem{proposition}[theorem]{Proposition}
\newtheorem{definition}[theorem]{Definition}
\newtheorem{conjecture}[theorem]{Conjecture}

\newtheorem{remark}[theorem]{Remark}
\newtheorem{rrule}[theorem]{Reduction Rule}

\crefname{rrule}{Reduction Rule}{Reduction Rules}
\crefname{problem}{Problem}{Problems}
\crefname{observation}{Observation}{Observations}
\crefname{construction}{Construction}{Constructions}
\crefname{figure}{Figure}{Figures}
\crefname{conjecture}{Conjecture}{Conjectures}

\newcommand{\RD}{$(\Rightarrow)$}
\newcommand{\LD}{$(\Leftarrow)$}
\date{}
\author{René van Bevern\\
  Department of Mathematics and Mechanics,
  Novosibirsk State University,
  Novosibirsk, Russian Federation,
  \texttt{rvb@nsu.ru}
  \and
  Till Fluschnik
  \\
  Technische Universit\"at Berlin,
  Faculty~IV,
  Algorithmics and Computational Complexity,
  Berlin, Germany,
  \texttt{till.fluschnik@tu-berlin.de}
  \and
  Oxana Yu.\ Tsidulko\\
  Department of Mathematics and Mechanics,
  Novosibirsk State University,
  Novosibirsk, Russian Federation,
  \\
  Sobolev Institute of Mathematics of the Siberian Branch of the Russian Academy of Sciences,
  Novosibirsk, Russian Federation,
  \texttt{tsidulko@math.nsc.ru}
}
\title{On approximate
  data reduction for
  the Rural Postman Problem:
  Theory and experiments\thanks{A preliminary version of this work appeared in the
    Proceedings of the 18th International Conference
    on Mathematical Optimization Theory and Operations Research
    (MOTOR 2019), Ekaterinburg, Russian Federation, July 8-12, 2019
    \citep{BFT19}.
    This work provides all proofs of the theorems stated in the conference version,
    a stronger version of \cref{thm:easy},
    WK[1]-completeness results (\cref{sec:hard}),
    and an experimental evaluation
    of our data reduction algorithm (\cref{sec:exp}).
  }
}

\begin{document}
\maketitle

\begin{abstract}
  \noindent
  \looseness=-1
  Given an undirected graph with edge weights
  and a subset~$R$ of its edges,
  the Rural Postman Problem (RPP)
  is to find a closed walk
  of minimum total weight
  containing all edges of~$R$.
  We prove that RPP is WK[1]-complete
  parameterized by
  the number and weight~$d$
  of edges
  traversed additionally to the required ones.
  Thus RPP instances cannot be polynomial\hyp time
  compressed to instances of size polynomial in~$d$
  unless the polynomial\hyp time hierarchy collapses.
  In contrast,
  denoting by~$b\leq 2d$ the number of vertices
  incident to an odd number of edges of~$R$
  and by~$c\leq d$ the number of connected
  components formed by the edges in~$R$,
  we show how to reduce
  any RPP instance~$I$
  to an
  RPP instance~$I'$ with
  \(2b+O(c/\varepsilon)\)~vertices
  in $O(n^3)$~time
  so that any \(\alpha\)-approximate solution
  for~$I'$
  gives an \(\alpha(1+\varepsilon)\)-approximate
  solution for~$I$,
  for any~$\alpha\geq 1$ and $\varepsilon>0$.
  That is,
  we provide a
  polynomial\hyp size approximate
  kernelization scheme (PSAKS).
  We experimentally evaluate it
  on wide-spread benchmark data sets
as well as on two real snow plowing instances from Berlin.
We also make first steps
  towards a PSAKS for the parameter~$c$.

\end{abstract}


\section{Introduction}

In the framework of lossy kernelization
\citep{LPRS17,FKRS18},
we study trade-offs between
the provable effect of data reduction
and
the provably achievable solution quality
for the following classical
vehicle routing problem \citep{Orl74}.

\pagebreak[3]
\begin{problem}[Rural Postman Problem, RPP]\leavevmode
  \begin{compactdesc}
    
  \item[Instance:] An undirected graph~$G=(V,E)$
    with $n$~vertices,
    edge weights~\(\cost\colon E\to\N\cup\{0\}\),
    and a multiset~$R$
    of \emph{required edges} of~$G$.
    
  \item[Task:] Find
    a closed walk~$W^*$ in~$G$
    containing each edge of~$R$
    and minimizing
    the total weight~$\cost(W^*)$
    of the edges on~$W^*$.
  \end{compactdesc}
\end{problem}
\noindent
We call any closed walk containing
each edge of~$R$ an \emph{\sol{}}.
By RPP we will also refer to
the \emph{decision problem}
where one additionally gets a non\hyp negative integer~$k\in\N$
in the input
and the task is to decide
whether there is an \sol{}~$W$
of weight~$\cost(W)\leq k$.

RPP has direct applications
in snow plowing,
street sweeping,
meter reading \citep{EGL95,CL15},
vehicle depot location \citep{GL99},
drilling, and plotting \citep{GJR91,GI01}.
The undirected version
occurs especially
in rural areas, where
service vehicles
can operate in both directions
even on one\hyp way roads \citep{EM91}.
Moreover,
RPP is a special case of
the Capacitated Arc Routing Problem (CARP)
\citep{GW81}
and used in all ``route first, cluster second'' algorithms
for CARP \citep{Ulu85,BBLP06,BE08},
which are notably the only ones
with proven
approximation guarantees \citep{Jan93,Woe08,BKS17}.
Improved approximations for RPP
automatically
lead to better approximations for CARP.

\looseness=-1
There is a folklore
polynomial\hyp time 3/2-approximation for RPP
based on the \citeauthor{Chr76}-\citeauthor{Ser78}
algorithm for the metric Traveling Salesman Problem \citep{Chr76,Ser78,BSxxb}
(we refer to \citet{EGL95} or \citet{BNSW15}
for a detailed algorithm description).
We aim for \((1+\varepsilon)\)\hyp approximations
for all~$\varepsilon>0$.
Unfortunately,
RPP contains the metric Traveling Salesman Problem
as a special case,
which cannot be
polynomial\hyp time
approximated within any factor smaller than 123/122
unless P${}={}$NP \citep{KLS15}.
Thus,
finding $(1+\varepsilon)$-approximations
even for constant small~$\varepsilon>0$
typically requires exponential time.
We present data reduction rules
for this task.
Their effectivity depends
on the desired~$\varepsilon$.

\subsection{Our contributions and outline of this paper}
In \cref{sec:prelims},
we introduce basic notation
of graph theory,
approximation algorithms,
parameterized complexity,
problem kernelization,
and WK[1]-completeness.
In \cref{sec:solstruct},
we prove basic
properties of optimal \sol{}s.

Using the recently introduced concept
of WK[1]-hardness \citep{HKS+15b},
in \cref{sec:hard},
we prove that it is hard
to reduce exactly solving RPP
to solving instances
of size polynomial in $\cost(W^*)-\cost(R)+|W^*|-|R|$,
which is
the weight and number
of the \emph{deadheading edges}
traversed additionally to the required ones:
\begin{theorem}
  \label{thm:wk1}
  RPP is WK[1]-complete
  parameterized by~$\Theta(\cost(W^*)-\cost(R)+|W^*|-|R|)$,
  where $W^*$~is an optimal \sol{}
  with a minimum amount of edges and
  WK[1]-hardness holds
  even in complete graphs
  with metric edge weights 1 and~2.
\end{theorem}
\noindent

\noindent
In contrast to \cref{thm:wk1},
in \cref{sec:psaks},
we show that effective data reduction for RPP
is possible
if one is interested in \((1+\varepsilon)\)-approximations.

\begin{theorem}
  \label{thm:psaks}
  There is an algorithm that,
  given \(\varepsilon>0\)
  and
  an RPP instance~$(G,R,\cost)$
  with $b$~vertices incident to an odd number of edges in~$R$
  and
  whose edges in~$R$ form $c$~connected components,
  reduces~$(G,R,\cost)$ to an RPP instance~$(G',R',\cost')$
  in $O(n^3+|R|)$~time
  such that
  \begin{compactenum}[(i)]
  \item\label{psaks1}
    the number of vertices in~$G'$ is $2b+O(c/\varepsilon)$,
    
  \item\label{psaks2}
    the number of required edges is~$|R'|\leq 4b+O(c/\varepsilon)$,
    
  \item\label{psaks3} the maximum edge weight
    with respect to~$\cost'$ is $O((b+c)/\varepsilon^2)$,
    and
  \item\label{psaks4}
    any $\alpha$-approximate solution
    for~$I'$
    for some \(\alpha\geq 1\)
    can be transformed into
    an $\alpha(1+\varepsilon)$-approximate solution
    for~$I$ in polynomial time.
  \end{compactenum}
\end{theorem}
\noindent
Finally,
in \cref{sec:exp},
we experimentally evaluate
the data reduction algorithm from \cref{thm:psaks}.

\paragraph{Discussion of our results.}
\cref{thm:wk1,thm:psaks} complement each other
since the number $|W^*|-|R|$ of deadheading arcs
is at least $\max\{b/2,c\}$
(see \cref{sec:bounds}).
Thus,
\cref{thm:wk1} shows that
it is hard to polynomial\hyp time reduce RPP to instances
of size $\poly(b+c)$ without
loss in the solution quality,
whereas \cref{thm:psaks} allows to do so
with arbitrarily small loss.
In \cref{sec:psaks-c},
we will also discuss difficulties
of getting rid of~$b$ or~$c$
in the size of the reduced instance~$I'$.

Notably,
the \(\alpha\)-approximate solution
for~$I'$ in \cref{thm:psaks}
may be obtained by any means,
for example exact algorithms or heuristics.
Thus, \cref{thm:psaks}
can be used to speed up expensive heuristics
without much loss in the solution quality.
In terms of
the recently introduced
concept of lossy kernelization
\citep{LPRS17},
\cref{thm:psaks}
yields a \emph{polynomial\hyp size approximate
  kernelization scheme (PSAKS)}.

In experiments,
on instances with few connected components,
the number of vertices and required edges is reduced
to about $50\,\%$
at an $1\,\%$ loss in the solution quality.
On real\hyp world snow plowing data from Berlin,
the number of vertices is reduced to
about $20\,\%$
without loss in the solution quality.

\subsection{Related work}
\paragraph{Classical complexity.}
RPP is strongly NP-hard \citep{LR76,Fre77},
its special case
with $R=E$ is the
polynomial\hyp time solvable
Chinese Postman problem
\citep{Chr73,EJ73}.
Containing the metric Traveling Salesman Problem
as a special case,
RPP is APX-hard \citep{KLS15}.
There is a folklore
polynomial\hyp time 3/2-approximation
based on the \citeauthor{Chr76}-\citeauthor{Ser78}
algorithm \citep{Chr76,Ser78,BSxxb} for the metric Traveling Salesman Problem 
(we refer to arc routing surveys \citep{EGL95,BNSW15}
for a detailed algorithmic description).
The Chinese Postman Problem
is equivalent to finding a minimum\hyp weight set of edges
whose addition makes a connected graph Eulerian \citep{EJ73,Chr73,Ser74}.
For a disconnected graph,
this is exactly RPP \citep{DMNW13,SBNW12,BNSW15}.

\paragraph{Parameterized complexity.}
\looseness=-1
\citet{DMNW13} showed an $O(4^d\cdot n^3)$-time
algorithm for the directed RPP,
where $d=|W^*|-|R|$~is the minimum number
of deadheading arcs
in an optimal solution~$W^*$.
It can be easily adapted to the undirected RPP.
\citet{SBNW11} showed an $O(4^{c\log b^2}\poly(n))$-time
algorithm for the directed RPP,
where
$c$~is the number of (weakly) connected components
induced by the required arcs in~$R$ and
$b=\sum_{v\in V}|\indeg(v)-\outdeg(v)|$.
It~is not obvious whether this algorithm
can be adapted to the undirected RPP
maintaining its running time.
\citet{GWY17} showed
a randomized algorithm
that solves the directed and undirected
RPP in $f(c)\poly(n)$~time
if edge weights are bounded polynomially in~$n$.
It is based on the Schwartz-Zippel Lemma \citep{Zip79,Sch80}
for randomized polynomial identity testing.
The existence of a deterministic algorithm
with this running time is open
\citep{SBNW12,BNSW15,GWY17}.

\paragraph{Exact kernelization.}
RPP can easily be reduced
to an equivalent instance with $2|R|$~vertices
\citep{BNSW15}.
By shrinking the weights using a theorem of \citet{FT87}
one gets a so\hyp called \emph{problem kernel}
of size polynomial
in the number of required edges
(we refer to \citet[Section~5.3]{BBF+xx} for details).
In contrast,
\citet{SBNW11} showed that,
unless the polynomial\hyp time
hierarchy collapses,
the directed RPP has no problem kernel
of size polynomial in the
number of deadheading arcs.
This result is strengthened by our \cref{thm:wk1},
which shows even WK[1]-hardness,
also of the directed RPP.

\paragraph{Lossy kernelization.}
\looseness=-1
Due to the kernelization hardness of many problems,
recently the concept of approximate kernelization
has gained increased interest \citep{FKRS18,LPRS17}.
In this context,
\citet{EHR17} called for finding
connectivity\hyp constrained problems
that do not have polynomial\hyp size kernels
but \(\alpha\)-approximate polynomial\hyp size kernels.
Our \cref{thm:wk1,thm:psaks} exhibit
that RPP is such a problem.
Among the so far few known lossy kernels
\citep{KMRT16,EHR17,EKM+18,KMR18,LPRS17},
our \cref{thm:psaks}
stands out since it shows a
time \emph{and} size efficient PSAKS,
which is a property previously observed
only in results of \citet{KMR18}.
Moreover,
\cref{thm:psaks}
is apparently the first lossy kernelization result
for parameters above lower bounds,
which previously got attention in exact kernelization.

\section{Preliminaries}
\label{sec:prelims}
\subsection{Set and graph theory}

\paragraph{Sets and multisets.}
By \(\N\) we denote the set of natural numbers
including zero.
For two multisets~$A$ and~$B$,
$A\uplus B$ is the multiset
obtained
by adding the multiplicities
of elements in~$A$ and~$B$.
By $A\setminus B$ we denote the multiset
obtained by subtracting
the multiplicities of elements in~$B$
from the multiplicities of elements in~$A$.
Finally,
given some weight function~\(\cost\colon A\to\N\),
the \emph{weight} of a multiset~$A$
is $\cost(A):=\sum_{e\in A}\mathbb \nu(e)\cost(e)$,
where $\mathbb \nu(e)$~is
the multiplicity of~$e$ in~$A$.

\paragraph{Graph theory.}
Graphs in our work
are allowed to have loops and parallel edges,
so that they are actually
\emph{multigraphs}~$G=(V,E)$
with a set~$V(G):=V$ of \emph{vertices},
a multiset~$E(G):=E$ over $\{\{u,v\}\mid u,v\in V\}$
of \emph{(undirected) edges},
and \emph{edge weights}~$\cost\colon E\to\N$.
Parallel edges in our graphs are indistinguishable
from each other and all have the same weight.
For a multiset~$R$ of edges,
we denote by~$V(R)$ the set of their incident vertices.

\paragraph{Paths and cycles.}
\looseness=-1
A \emph{walk from~$v_0$ to~$v_\ell$ in~$G$} is a
sequence~$w=(v_0,e_1,v_1,e_2,v_2,\dots,\allowbreak
e_\ell,v_\ell)$
such that $e_i$~is an edge
with end points~$v_{i-1}$ and~$v_{i}$
for each~$i\in\{1,\dots,\ell\}$.
If $v_0=v_\ell$,
then we call~$w$ a \emph{closed walk}.
If all vertices on~$w$ are pairwise distinct,
then $w$~is a \emph{path}.
If only its first and last vertex coincide,
then $w$~is a \emph{cycle}.
By $E(w)$~we denote the multiset of edges on~$w$,
that is,
each edge appears on~$w$ and in~$E(w)$ equally often.
The \emph{length} of walk~$w$
is its number $|w|:=\ell=|E(w)|$ of edges.
The \emph{weight} of walk~$w$
is $\cost(w):=\sum_{i=1}^\ell\cost(e_\ell)$.
An \emph{Euler tour for~$G$}
is a closed walk that traverses
each edge of~$G$
exactly as often as it is present in~$G$.
A graph is \emph{Eulerian}
if it allows for an Euler tour.

\paragraph{Connectivity and blocks.}
\looseness=-1
Two vertices~$u,v$ of~$G$ are \emph{connected} if there is a path from~$u$ to~$v$ in~$G$.
A \emph{connected component} of~$G$
is a maximal subgraph of~$G$
in which the vertices are mutually connected.
A vertex~$v$ of~$G$ is a
\emph{cut vertex}
if removing~$v$
and its incident edges
increases the number of
connected components of~$G$.
A~\emph{biconnected component}
or \emph{block} of~$G$
is a maximal subgraph
without cut vertices.

\paragraph{Edge- and vertex\hyp induced subgraphs.}
For a subset~$U\subseteq V$ of vertices,
\emph{the subgraph~$G[U]$ of~$G=(V,E)$ induced by~$U$}
consists of the vertices of~$U$
and all edges of~$G$ between them
(respecting multiplicities).
For a multiset~$R$ of edges of~$G$,
$\eG G R:=(V(R),R)$
is the graph
\emph{induced by the edges in~$R$}.
For a walk~$w$,
we also denote $\eG{G}{w}:=\eG{G}{E(w)}$.
Note that $\eG{G}{R}$ and~$\eG{G}{w}$
do not contain isolated vertices yet
might contain
edges with a higher multiplicity than~$G$ and,
therefore,
are not necessarily sub(multi)graphs of~$G$.

\subsection{Decision problems, optimization problems, approximation}

\begin{definition}
  A \emph{decision problem} is a subset~$\Pi\subseteq\Sigma^*$
  for some finite alphabet~$\Sigma$.
  The task is,
  given an \emph{instance}~$x\in\Sigma^*$,
  determining whether $x\in \Pi$.
  If $x\in\Pi$,
  then $x$~is a \emph{yes\hyp instance}.
  Otherwise,
  it is a \emph{no\hyp instance}.
\end{definition}

\noindent
For optimization problems,
we use the terminology of \citet{GJ79}.
We will only consider \emph{minimization problems} in our work.
\begin{definition}
  An \emph{combinatorial optimization problem}~$\Pi$ is a triple~$\Pi=(D_\Pi,S_\Pi,m_\Pi)$,
  where
  \begin{compactenum}
  \item $D_\Pi$ is a set of \emph{instances},
  \item $S_\Pi$ is a function
    assigning to each instance $I\in D_\Pi$
    a finite set~$S_\Pi(I)$ of
    \emph{(feasible) solutions}, and
  \item $m_\Pi$ is a function
    that assigns a \emph{solution cost}~$m_\Pi(I,\sigma)$
    to each feasible solution~$\sigma\in S_\Pi(I)$
    of an instance~$I\in D_\Pi$.
  \end{compactenum}
  An \emph{optimal solution}
  for an instance~$I\in D_\Pi$
  is a feasible solution~$\sigma\in S_\Pi(I)$
  minimizing~$m_\Pi(I,\sigma)$.
  Its cost is denoted as $\OPT_\Pi(I)$,
  where we drop the subscript~$\Pi$
  when the optimization problem is clear from context.
\end{definition}

\noindent
If the set~$S_\Pi(I)$ for any instance~$I\in D_\Pi$
is polynomial\hyp time recognizable,
any solution~$\sigma\in S_\Pi(I)$ has size~$|\sigma|\leq\poly(|I|)$,
and the function~$m_\Pi$ is polynomial\hyp time computable,
then one also calls~$\Pi$ an \emph{NP-optimization problem}.

\begin{definition}
  An \emph{$\alpha$\hyp approximate solution} for an instance~$I$
  of a combinatorial optimization problem~$\Pi$
  is a feasible solution of cost at most $\alpha\cdot\OPT_\Pi(I)$.
\end{definition}

\subsection{Kernelization}

\looseness=-1
Kernelization is the main notion
of data reduction with provable performance guarantees
\citep{FLSZ19}.
Since proving that a polynomial\hyp time algorithm
always shrinks the input instance of an NP-hard problem,
say from size~$n$ to $n-1$, would imply P${}={}$NP,
the size of the reduced instance is measured
in dependence of a
\emph{parameter} of the input instance.

We formalize parameters using the terminology of \citet{FG06},
since it allows to parameterize decision and optimization problems
in a uniform way:

\begin{definition}
  A \emph{parameterization} is a
  polynomial\hyp time computable
  mapping~$\kappa\colon\Sigma^*\to\N$
  of instances (of decision or optimization problems)
  to a \emph{parameter}.
  For a (decision or optimization) problem~$\Pi$
  and parameterization~$\kappa$,
  $(\Pi,\kappa)$ is called
  a \emph{parameterized (decision or optimization) problem}.
\end{definition}

\begin{definition}
 \label{def:compression}
 A~\emph{kernelization} for a parameterized decision problem~$(\Pi,\kappa)$
 is a polynomial\hyp time algorithm
 that maps any instance~$x\in\Sigma^*$
 to an instance~$x'\in\Sigma^*$
 such that
 \begin{compactenum}[(i)]
  \item $x\in \Pi \iff x'\in \Pi$, and
  \item $|x'|\leq g(\kappa(x))$ for some computable function~\(g\).
  \end{compactenum}
  We call \(x'\) the \emph{problem kernel}
  and \(g\) its \emph{size}.
\end{definition}

\noindent
A generalization of problem kernels
are \emph{Turing kernels},
where one is allowed to generate multiple reduced
instances instead of a single one.

\begin{definition}
  A~\emph{Turing kernelization} for
  a parameterized decision problem~$(\Pi,\kappa)$
  is an algorithm~$A$ that decides~$x\in \Pi$
  in polynomial time
  given access to an oracle
  that answers $x'\in \Pi$ in constant time
  for any \(x'\in\Sigma^*\) with
  $|x'|\leq g(\kappa(x))$,
  where $g$~is an arbitrary function
  called the \emph{size} of the Turing kernel.
\end{definition}

\noindent
Since \cref{thm:wk1} means that
it is hard to obtain problem kernels for RPP
even with size polynomial in a relatively large parameter,
we will consider approximate problem kernels
\citep{LPRS17}:

\begin{definition}
  A \emph{$\beta$-approximate kernelization}
  for a parameterized optimization problem~$(\Pi,\kappa)$
  consists of two polynomial\hyp time algorithms:
  \begin{compactenum}[(i)]
  \item The first algorithm reduces an instance~$I$
    of~$\Pi$ to an instance~$I'$ of~$\Pi$
    such that $|I'|\leq g(\kappa(I))$
    for some computable function~$g\colon\N\to\N$.
  \item The second algorithm turns any
    \(\alpha\)\hyp approximate solution for~$I'$
    into an \(\alpha\beta\)\hyp
    approximate
    solution for~$I$.
  \end{compactenum}
  We call $g$ the \emph{size} of the \emph{approximate kernel}~$I'$.
\end{definition}

\noindent
In fact, we show polynomial\hyp size
approximate kernelization schemes
\citep{LPRS17}:

\begin{definition}
  A \emph{polynomial\hyp size approximate
    kernelization scheme (PSAKS)}
is a family of
  $(1+\varepsilon)$-approximate kernelizations
  yielding approximate kernels of polynomial size for every fixed~$\varepsilon>0$.
\end{definition}

\subsubsection{Kernelization hardness}
\emph{WK[1]-complete} parameterized decision problems
do not have problem kernels of polynomial size
unless the polynomial\hyp time hierarchy collapses
and are conjectured not to have
Turing kernels of polynomial size either~\citep{HKS+15b}.
An archetypal WK[1]-complete
problem is the following~\citep{HKS+15b}:
\begin{problem}[NDTM Halting]\leavevmode
  \label{prob:ndtm}
  \begin{compactdesc}
  \item[Instance:] A nondeterministic Turing machine~$\M$ and an integer~$t$.
  \item[Parameter:] $t\log |\M|$.
  \item[Question:] Does~$\M$ halt in~$t$
    steps on the empty input string?
  \end{compactdesc}
\end{problem}

\noindent
The class WK[1] can now be defined
as the class of all parameterized problems
reducible to NDTM Halting using
the following type of reduction.

\begin{definition}
  A \emph{polynomial parameter transformation (PPT)}
  of a parameterized decision problem~$(\Pi,\kappa)$
  into a parameterized decision problem~$(\Pi',\kappa')$
  is an algorithm that
  maps any instance~$x\in\Sigma^*$
  to an instance~$x'\in\Sigma^*$
  in polynomial time
  so that
  \begin{compactenum}[(i)]
  \item $x\in \Pi\iff x'\in \Pi'$ and
  \item $\kappa'(x')\in \poly(\kappa(x))$.
 \end{compactenum}
\end{definition}
\begin{definition}
  WK[1] is the class of parameterized decision problems
  PPT-reducible to NDTM Halting.
  A~parameterized decision problem~$(\Pi,\kappa)$ is \emph{WK[1]-hard}
  if every parameterized decision problem in WK[1]
  is PPT-reducible to~$(\Pi,\kappa)$.
  It is \emph{WK[1]-complete}
  if it is WK[1]-hard and contained in WK[1].
\end{definition}
\noindent
Notably,
since PPT-reducibility is a transitive relation,
to prove WK[1]-hardness of a parameterized decision problem,
it is enough to PPT-reduce any other WK[1]-hard parameterized decision
problem
to it.

\subsection{Approximate weight reduction}

We will use the following lemma
to shrink edge weights
so that their encoding length will be
polynomial in the number of vertices and edges of the graph.
It is a generalization
of an idea implicitly
used for weight reduction
in a proof of \citet[Theorem~4.2]{LPRS17}
and shrinks weights faster
and more significantly
than a theorem of \citet{FT87}
that is frequently used
in the exact kernelization of weighted
problems \citep{BFT20,EKMR17,MV15,BBF+xx}.
We first state the lemma,
and thereafter intuitively describe its application to RPP.

\begin{lemma}[lossy weight reduction]
  \label{lem:looseweight}
  Let $\mathcal F\subseteq\Q^n_{\geq 0}$ and $\cost\in\mathbb \Q_{\geq 0}^n$ be
such that
  \begin{compactitem}
  \item \(\|\cost\|_{\infty}\leq\pWR{}\) for some \(\pWR{}\in \Q\) and
  \item \(\|x\|_1\leq N\) for some \(N\in \N\) and all $x\in\mathcal F$.
  \end{compactitem}
  Let
  \(x^*\in\argmin\{\cost^\top x\mid x\in\mathcal F\}\)
  and
  \(\bar x^*\in\argmin\{\bar \cost^\top x\mid x\in\mathcal F\}\).
  Then, for any~$\varepsilon>0$,
  using $O(n)$ arithmetic operations involving only the numbers~$\varepsilon$,
  $\beta$, and the components of $\cost$,
  one can compute $\bar \cost\in \N^n$ such that
  \begin{compactenum}[(i)]
  \item\label{wr1} \(\|\bar \cost\|_\infty\leq N/\varepsilon\) and
  \item\label{wr2} for any \(\apxfakt{}\in\Q\) and
 \(x\in\mathcal F\) with
    \(\bar \cost^\top x\leq \apxfakt{}\cdot \bar \cost^\top\bar x^*\),
    one has
    \(
    		\cost^\top x\leq \apxfakt{}\cdot \cost^\top x^*+\varepsilon\pWR{}.
    \)
\end{compactenum}
\end{lemma}

\noindent
Note that one can easily prove a version of \cref{lem:looseweight}
for maximization problems.
To apply \cref{lem:looseweight} to RPP,
we will take $\mathcal F$~to be the set
of inclusion\hyp minimal \sol{}s
(encoded as vectors~$x\in\mathcal F$
having an entry for each edge
that specifies how often it is included in the \sol{}),
and $\cost$~to be a vector
having an entry for each edge specifying its weight.
Then the linear forms $\cost^\top x$ and $\bar\cost^\top x$
give the weight of the \sol{} encoded by~$x$
with respect to the initial weights~$\cost$
and the reduced weights~$\bar\cost$,
respectively.
The linear forms occurring in the lemma
seem to limit it to problems with
linear or additive goal functions,
yet in fact are powerful enough
to model many non\hyp additive goal functions
as well
\citep{BBF+xx}.

\begin{proof}[Proof of \cref{lem:looseweight}]
  Choose $M=(\varepsilon \pWR{})/N$
  and
  $\bar{\cost}_i = \lfloor \cost_i/M \rfloor$
  for each~$i\in\{1,\ldots,n\}$.
  Since $\cost_i\geq 0$ for each $i\in\{1,\ldots,n\}$,
  we have $\bar{\cost}\in \N^n$.
  Moreover,
  due to $\|\cost\|_\infty \leq \pWR{}$,
  for each~$i\in\{1,\ldots,n\}$,
  we have~$\bar{\cost}_i \leq \pWR{}/M = N/\varepsilon$,
  proving \eqref{wr1}.
  
  To prove \eqref{wr2},
  let $x\in \mathcal F$~be
  such that $\bar \cost^\top x\leq \apxfakt{}\cdot \bar \cost^\top\bar x^*$.
  By the choice of $\bar{\cost}$
  for each~$i\in\{1,\ldots,n\}$,
  we have~$\cost_i\leq M\cdot (\bar{\cost}_i+1)$.
  Moreover, we have
  \begin{align}\label{eq-w-shrinked}
    M \cdot \bar{\cost}^\top x&\leq \cost^\top x \qquad \text{ for all } x\in \mathcal F.
  \end{align} 
  It follows that
  \(\cost^\top x  \leq M\cdot (\bar{\cost}^\top x + \|x\|_1) 
    \leq M\cdot \bar{\cost}^\top x + \varepsilon\pWR{} 
    \leq \apxfakt{}\cdot M\cdot \bar{\cost}^\top \bar x^* + \varepsilon\pWR{}.\)
    By (\ref{eq-w-shrinked})
  and the choice of $\bar{x}^*$,
  we have
  \(\bar{\cost}^\top \bar{x}^* \le \bar{\cost}^\top x^* \le \cost^\top x^*/M.\)
  Finally,
  $\cost^\top x \le \apxfakt{}\cdot M\cdot \bar{\cost}^\top \bar x^* + \varepsilon\pWR{} \le \apxfakt{}
  \cdot \cost^\top x^* + \varepsilon\pWR{}.$
\end{proof}

\section{Solution structure}
\label{sec:solstruct}
In this section,
we prove fundamental properties
of optimal \sol{}s.
To make these hold,
we first establish the triangle inequality in \cref{sec:metric}.
In \cref{sec:eeprop},
we translate RPP to the problem
of finding \emph{Eulerian extensions}.
In \cref{sec:bounds},
we derive inequalities
to bound parts of
optimal \sol{}s.

\subsection{Triangle inequality}
\label{sec:metric}
Without loss of generality,
we will assume that the weight function
satisfies the triangle inequality:

\begin{proposition}[\cite{BHNS14}]
  \label{lem:triangle}
  \looseness=-1
  In $O(n^3)$~time,
  an RPP instance $(G,R,\cost)$
  can be turned into an RPP instance $(G',R,\cost')$
  such that
\begin{compactenum}
  \item $G'$~is a complete graph,
  \item $\cost'$~satisfies the triangle inequality, and
  \item any \(\alpha\)-approximate solution
  for~$(G,R,\cost')$
  can be turned into an
  \(\alpha\)-approximate solution for~$(G,R,\cost)$
  in polynomial time. 
  \end{compactenum}
\end{proposition}

\begin{remark}
  \label{rem:metric-rpp}
  \looseness=-1
  \cref{lem:triangle} holds
  in particular for \(\alpha=1\)
  and
  does not increase
the number of connected components of~$\eG{G}{R}$,
  the number of odd\hyp degree vertices of~$\eG{G}{R}$,
  the number and weight of deadheading edges
  of an optimal \sol{}.
  Thus,
  it is sufficient to prove
  \cref{thm:wk1,thm:psaks}
  for RPP with triangle inequality.
  We will henceforth assume that the input graph is complete
  and satisfies the triangle inequality.
\end{remark}

\subsection{Edge-minimizing Eulerian extensions}
\label{sec:eeprop}
\noindent
\looseness=-1
Consider any \sol{}~$W$
for an RPP instance~$(G,R,\cost)$.
Then $\eG{G}{W}$~is an Eulerian
supergraph of~$\eG{G}{R}$
whose total edge weight is~\(\cost(W)\).
Moreover,
any Eulerian supergraph~$\eG{G}{W'}$
of~$\eG{G}{R}$
yields an \sol{} for~$(G,R,\cost)$
of total weight~\(\cost(W')\).
Thus,
\sol{}s one-to-one
correspond to \emph{\EE{}s}
\citep{DMNW13}:\footnote{Before, this correspondence was
  observed for the Chinese Postman Problem
  independently by \citet{Chr73,EJ73}, and \citet{Ser74}.}

\begin{definition}An \emph{\EE{}} (EE)
  for an RPP instance~$(G,R,\cost)$
  is a multiset~$S$
  of edges such that
  \(\eG{G}{R\uplus S}\) is Eulerian.
  We say that an \EE{}~$S$ is
  \emph{\emi{}}
  if there is no \EE{}~$S'$
  with \(|S'|<|S|\) and \(\cost(S')\leq \cost(S)\).
\end{definition}

\noindent
In the following,
we will concentrate on finding minimum\hyp weight
\EE{}s rather than \sol{}s
and exploit that a graph without
isolated vertices is Eulerian
if and only if it is connected and \emph{balanced}:
\begin{definition}
  A vertex is \emph{balanced}
  if it has even degree.
  A graph is \emph{balanced}
  if each of its vertices is balanced.
\end{definition}

\noindent
Thus,
solving RPP reduces to finding
a minimum\hyp weight set~$S$ of edges
such that~$\eG{G}{R\uplus S}$ is connected and balanced.
Since an Euler tour
in the Eulerian graph~$\eG{G}{R\uplus S}$
is computable in linear time
using Hierholzer's algorithm \citep{HW73,Fle91},
we can easily recover an \sol{} from an \EE{}.
\begin{proposition}
  \label{prop:linearlift}
  Let $(G,R,\cost)$~be an RPP instance.

  \begin{compactenum}[(i)]
  \item From any \sol{}~$W$ for~$(G,R,\cost)$,
  one can compute an \EE{}~$S$
  of weight \(\cost(S)=\cost(W)-\cost(R)\)
  in time linear in~$|W|$.
  
\item   From any \EE{}~$S$ for~$(G,R,\cost)$,
  one can compute an \sol{}~$W$
  of weight \(\cost(W)=\cost(R)+\cost(S)\)
  in time
  linear in~$|R|+|S|$.
\end{compactenum}
\end{proposition}
\noindent
Assuming the triangle inequality,
any \sol{}
can be shortcut 
to contain only vertices incident to required edges.
\begin{observation}
  \label{obs:allR}
  Any \emi{} \EE{}~$S$ for an RPP instance~$(G,R,\cost)$
  satisfies $V(S)\subseteq V(R)$.
\end{observation}

\noindent
The following lemma,
in particular,
shows that no \emi{} \EE{}
contains required edges
between balanced vertices.
\begin{lemma}
  \label{lem:onebal}
  An \emi{} \EE{}~$S$
  for an RPP instance~$(G,\allowbreak
  R,\cost)$
  does not contain any
  edge~\(\{u,v\}\)
  such that $u$ and $v$ belong
  to the same connected component of~\(\eG{G}{R}\)
  and such that \(u\)~is \ba{} in~\(\eG{G}{R}\).
\end{lemma}
\begin{proof}
  \looseness=-1
  Towards a contradiction,
  assume that \(\{u,v\}\in S\).
Since $u$ is
  \ba{} in~$\eG{G}{R}$ and~$\eG{G}{R\uplus S}$,
  \(S\)~additionally contains
  an edge~$\{u,w\}$
  (possibly, $v=w$).
Then
  $(S'\setminus\{\{u,v\},\{u,w\}\})\uplus\{\{v,w\}\}$
  satisfies \(|S'|<|S|\)
  and also is an \EE{}:
  the balance of $u$, $v$ and $w$
  is the same in $\eG{G}{R\uplus S}$ and $\eG{G}{R\uplus S'}$,
  and $u$~still
  is connected to~$v$ in~$\eG{G}{R\uplus S'}$
  since $u$ and~$v$ belong
  to the same connected component of~$\eG{G}{R}$.
  Finally,
  using the triangle inequality, \(\cost(S')\leq \cost(S)\),
  contradicting the fact that
  $S$~is \emi{}.
\end{proof}

\begin{lemma}
  \label{lem:2C-2}
  \looseness=-1
  Let $(G,R,\cost)$~be an RPP instance
  and $c$~be the number of connected components
  of~$\eG{G}{R}$.
  At most $2c-2$ \ba{} vertices
  in~$\eG{G}{R}$
  are incident to edges of
  an \emi{} \EE{}
  and this bound is tight.
\end{lemma}
\begin{proof}
  Let \(S\)~be an \emi{} \EE{} for~$(G,R,\cost)$
  and $T\subseteq S$ be an inclusion\hyp minimal
  subset such that $\eG{G}{R\uplus T}$~is connected.
  Then $|T|=c-1$
  and $S\setminus T$~is
  an \emi{} \EE{} for~$(G,R\uplus T,\cost)$.
  Thus,
  by \cref{obs:allR},
  $V(S\setminus T)\subseteq V(R\uplus T)$.
  Combining this with
  \cref{lem:onebal},
  $S\setminus T$~does not contain
  any edges incident to \ba{} vertices
  of~$\eG{G}{R\uplus T}$.
  The only vertices that might
  be \ba{} in~$\eG{G}{R}$
  but not in~$\eG{G}{R\uplus T}$
  are the at most \(2c-2\)~end points
  of edges in~$T$.
  In the worst case,
  all of them are incident to edges in~$S$.
  \cref{fig:2c-2-tight}
  shows that the bound is tight.
  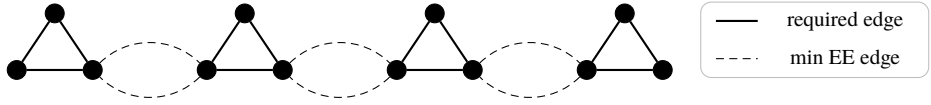
\begin{figure}[t]
  \centering
  \begin{tikzpicture}

    \def\xr{1}
    \def\yr{1}
    \tikzstyle{xnode}=[circle, fill, scale=3/4, draw]
    \tikzstyle{redge}=[thick, black]
    \tikzstyle{sedge}=[densely dashed, black]

    \begin{scope}[]
    \node (a) at (0,0)[xnode]{};
    \node (b) at (0.5*\xr,0.75*\yr)[xnode]{};
    \node (c) at (1*\xr,0)[xnode]{};
    \draw[redge] (a) -- (b) -- (c) -- (a);
    \end{scope}

    \begin{scope}[xshift=2.5*\xr cm]
    \node (ax) at (0,0)[xnode]{};
    \node (bx) at (0.5*\xr,0.75*\yr)[xnode]{};
    \node (cx) at (1*\xr,0)[xnode]{};
    \draw[redge] (ax) -- (bx) -- (cx) -- (ax);
    \end{scope}

    \draw[sedge] (c) to [out=45,in=135](ax);
    \draw[sedge] (ax) to [out=-135,in=-45](c);

    \begin{scope}[xshift=5*\xr cm]
    \node (a) at (0,0)[xnode]{};
    \node (b) at (0.5*\xr,0.75*\yr)[xnode]{};
    \node (c) at (1*\xr,0)[xnode]{};
    \draw[redge] (a) -- (b) -- (c) -- (a);
    \end{scope}

    \draw[sedge] (cx) to [out=45,in=135](a);
    \draw[sedge] (a) to [out=-135,in=-45](cx);

    \begin{scope}[xshift=7.5*\xr cm]
    \node (ax) at (0,0)[xnode]{};
    \node (bx) at (0.5*\xr,0.75*\yr)[xnode]{};
    \node (cx) at (1*\xr,0)[xnode]{};
    \draw[redge] (ax) -- (bx) -- (cx) -- (ax);
    \end{scope}

    \draw[sedge] (c) to [out=45,in=135](ax);
    \draw[sedge] (ax) to [out=-135,in=-45](c);

    \begin{scope}[xshift=9*\xr cm, yshift=-0.1*\yr cm]
      \draw[lightgray, rounded corners] (0,0) rectangle (3*\xr,1*\yr);
      \node at (1.9*\xr,0.75*\yr)[scale=0.8]{required edge};
    \draw[redge] (0.2*\xr,0.75*\yr) -- (0.75*\xr,0.75*\yr);
      \node at (1.9*\xr,0.25*\yr)[align=left,scale=0.8]{min EE edge};
    \draw[sedge] (0.2*\xr,0.25*\yr) -- (0.75*\xr,0.25*\yr);
    \end{scope}

    \end{tikzpicture}
    \caption{Proof that the bound given in \cref{lem:2C-2} is tight:
      $\eG{G}{R}$~has $c=4$~connected
      components
      and $2c-2=6$~vertices
      are incident to the \EE{}.}
  \label{fig:2c-2-tight}
\end{figure}
\end{proof}

\begin{remark}
  \label{rem:one-or-two}
  The following lemma shows that
  an \emi{} \EE{}
  contains exactly one edge
  incident to each unbalanced vertex of~$\eG{G}{R}$
  and either no or two edges
  incident to each balanced vertex of~$\eG{G}{R}$.
\end{remark}

\begin{lemma}
  \label{lem:at-most-two}
  Each vertex~$v\in V$
  is incident to at most two edges of
  an \emi{} \EE{}~$S$
  for an RPP instance~$(G,R,\cost)$.
\end{lemma}
\begin{proof}
  Towards a contradiction,
  assume that $S$~contains \(e_i=\{u_i,v\}\)
  for \(i\in\{1,2,3\}\).
  Obviously,
  \(S'=(S\setminus\{e_1,e_2\})\uplus\{\{u_1,u_2\}\}\)
  satisfies $|S'|<|S|$.
  Moreover,
  $\cost(S')\leq \cost(S)$~follows from the triangle inequality.
  We argue that $S'$~is an \EE{},
  contradicting the choice of~$S$.
  The proof is illustrated in \cref{fig:at-most-two}.
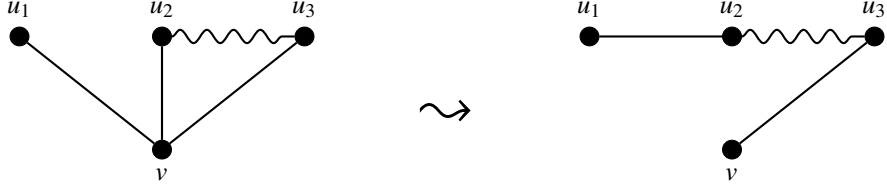
\begin{figure}[t]
    \centering
\begin{tikzpicture}
\def\xr{1.25}
      \def\yr{1}
      \tikzstyle{xnode}+=[circle, fill, scale=3/4, draw]
      \tikzstyle{xedge}+=[-,thick]
      \tikzstyle{xxedge}+=[-,thick, decorate, decoration=snake]

      \begin{scope}
        \node (v) at (0,0)[xnode,label=-90:{$v$}]{};
        \node (u1) at (-1.5*\xr,1.5*\yr)[xnode,label=90:{$u_1$}]{};
        \node (u2) at (0*\xr,1.5*\yr)[xnode,label=90:{$u_2$}]{};
        \node (u3) at (1.5*\xr,1.5*\yr)[xnode,label=90:{$u_3$}]{};
        \draw[xedge] (u1) to (v) to (u2);
        \draw[xedge] (v) to (u3);
        \draw[xxedge] (u2) to (u3);
      \end{scope}

      \node at (3*\xr,0.5*\yr)[scale=2]{$\leadsto$};

      \begin{scope}[xshift=6*\xr cm]
        \node (v) at (0,0)[xnode,label=-90:{$v$}]{};
        \node (u1) at (-1.5*\xr,1.5*\yr)[xnode,label=90:{$u_1$}]{};
        \node (u2) at (0*\xr,1.5*\yr)[xnode,label=90:{$u_2$}]{};
        \node (u3) at (1.5*\xr,1.5*\yr)[xnode,label=90:{$u_3$}]{};
        \draw[xedge] (u1) to (u2);
        \draw[xedge] (v) to (u3);
        \draw[xxedge] (u2) to (u3);
      \end{scope}
      
    \end{tikzpicture}
    \caption{Illustration of the proof of \cref{lem:at-most-two}. The wavy edge is a~$u_2$-$u_3$ path.}
    \label{fig:at-most-two}
  \end{figure}

  The balance of \(v\), \(u_1\), \(u_2\), and $u_3$
  is the same in \(\eG{G}{R\uplus S}\) and \(\eG{G}{R\uplus S'}\).
  It remains to show that
  the vertices~$v,u_1,u_2,u_3$ are connected
  in $\eG{G}{R\uplus S'}$.
  To this end,
  observe that
  $\eG{G}{R\uplus S}$ is Eulerian
  and thus contains two edge\hyp disjoint
  paths between~$u_2$ and~$u_3$.
  At most one of these paths
  contains the edge~$e_2$ and
  is lost in~$\eG{G}{R\uplus S'}$.
  Thus,
  $\eG{G}{R\uplus S'}$~contains the
  edges~$\{v,u_3\}$, $\{u_1,u_2\}$,
  and a path between $u_2$ and~$u_3$.
\end{proof}

\subsection{Inequalities}
\label{sec:bounds}

\begin{definition}
  \label{def:bounds}
  In the context of an RPP instance~$(G,R,\cost)$,
  we denote by
\begin{compactitem}[$W^*$ --]
\item[$R$ --] the set of required edges,

\item[$c$ --] the number of connected components in~$\eG{G}{R}$,
  
\item[$b$ --] the number of imbalanced vertices in~$\eG{G}{R}$,
  
\item[$W^*$ --] a minimum\hyp weight \sol{}
  with a minimum number of edges,

\item[$D$ --] a minimum\hyp weight \emi{}
  \EE{} for~$(G,R,\cost)$,

\item[$T$ --] a minimum\hyp weight
  set of edges such that~$\eG{G}{R\uplus T}$
  is connected,
  of minimum cardinality,
  and
 
\item[$M$ --] a minimum\hyp weight
  set of edges such that~$\eG{G}{R\uplus M}$
  is balanced,
  of minimum cardinality.
\end{compactitem}
\end{definition}
\begin{lemma}
  \label{lem:bounds}
  The following relations hold:

  \noindent
  \begin{minipage}{0.45\linewidth}
  \begin{align}
    \cost(W^*)&=\cost(R)+\cost(D),
                \label{cW=R+D}\\
    \cost(M)&\leq\cost(D),
              \label{cMD}
    \\
    \cost(T)&\leq\cost(D),
              \label{cTD}\\
    \cost(D)&\leq\cost(M)+2\cost(T),
              \label{cM+2T}
  \end{align}
\end{minipage}
\hfill
\begin{minipage}{0.45\linewidth}
  \begin{align}
    |W^*|&=|R|+|D|,
           \label{W=R+D}\\
    2b=|M|&\leq|D|,
            \label{MD}\\
    c-1=|T|&\leq|D|,
             \label{TD}\\
    |D|&\leq|M|+2|T|,
          \label{M+2T}
  \end{align}
\end{minipage}

\medskip
\noindent
where $|S|\leq|M|+2|T|$
holds for \emph{any} \emi{} \EE{}~$S$.
\end{lemma}
\begin{proof}
Equations
\eqref{cW=R+D} and \eqref{W=R+D} follow from
\cref{prop:linearlift}.
Inequalities
\eqref{cMD} and \eqref{MD}
follow by choice of~$M$
and the fact that,
since we assume the triangle inequality,
$M$~is simply a
minimum\hyp weight
perfect matching on the $b$~\imba{}
vertices in~$\eG{G}{R}$ \citep{EJ73,Chr73,Ser74}.
Inequalities
\eqref{cTD} and \eqref{TD}
follow by choice of~$T$.
Inequality
\eqref{cM+2T}
follows from the fact that
$\eG{G}{R\uplus M}$~is balanced
and adding each edge of~$T$ twice to it
does not change the balance of vertices,
yet connects the graph.
We now derive inequality~\eqref{M+2T}.
Consider any \emi{} \EE{}~$S$.
By \cref{lem:2C-2},
  a set~$X$ of at most $2c-2$~balanced vertices in~$\eG{G}{R}$
  are incident to edges of~$S$.
By \cref{rem:one-or-two},
  $S$~contains exactly one edge
  incident to each imbalanced vertex in~$\eG{G}{R}$
  and exactly two edges
  incident to each vertex in~$X$.
Thus,
by the handshaking lemma,
we get
\(
2|X|+b=2|S|.
\)
Therefore,
$|S|=|X|+b/2\leq 2c-2+|M|=2|T|+|M|$.
\end{proof}

\section{Hardness of kernelization}
\label{sec:hard}
\label{sec:wk1}
In this section,
we prove \cref{thm:wk1}.
We first show WK[1]-hardness in \cref{lem:wk1-hard},
then we show containment in WK[1] in
\cref{lem:wk1-containment}.
\cref{thm:wk1}
immediately follows from
\cref{lem:wk1-hard,lem:wk1-containment}
using \eqref{cW=R+D} and \eqref{W=R+D}.

\cref{lem:wk1-hard} means that,
according to the conjecture of \citet{HKS+15b},
RPP has no Turing kernels with size polynomial
in the number and weight of deadheading edges
in an optimal \sol{}.

\begin{lemma}
  \label{lem:wk1-hard}
  RPP is WK[1]-hard
  parameterized by
  $|T|+|M|+\cost(T)+\cost(M)\in\Theta(|D|+\cost(D))$
  even in complete graphs
  with metric edge weights one and two.
\end{lemma}

\noindent
To prove \cref{lem:wk1-hard},
we provide a
polynomial parameter transformation
from the following
known WK[1]-complete parameterized
problem \citep{HKS+15b}.

\begin{problem}[Multicolored Cycle]\leavevmode
  \begin{compactdesc}
  \item[Instance:] An undirected 
    graph~$G=(V,E)$
    with a
    vertex coloring \(c\colon V\to\{1,\dots,k\}\).
  \item[Parameter:] $k$.
  \item[Question:] Is there a cycle in~$G$
    containing exactly one vertex
    of each color?
  \end{compactdesc}
\end{problem}
\begin{proof}[Proof of \cref{lem:wk1-hard}]
  Let $I:=(G,c)$
  with a graph~$G=(V,E)$
  and a vertex $k$\hyp coloring~$c\colon V\to\{1,\dots,k\}$
  be an instance
  of Multicolored Cycle.
  For $i\in\{1,\dots,k\}$,
  we denote by $V_i:=\{v\in V\mid c(v)=i\}$
  the vertices of color~$i$.
  Now, consider the RPP instance~$I'=(G', R, \cost)$,
  illustrated in \cref{fig:wk1-hard}:
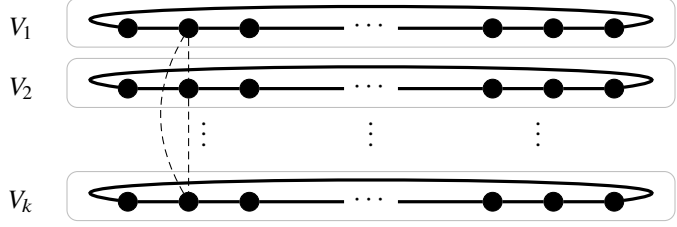
\begin{figure}[t]
    \centering
    \begin{tikzpicture}
      \def\xr{0.8}
      \tikzstyle{xnode}+=[circle, fill, scale=3/4, draw]
      \tikzstyle{redge}=[-,very thick,color=black];
      \tikzstyle{eeedge}=[-,densely dashed, black]
      \newcommand{\insta}[3]{
        \draw[rounded corners, lightgray] (0,#1-0.25) rectangle (10*\xr,#1+0.4);
      \node at (-0.75*\xr,#1)[]{#3};
        \node (a1#2) at (1*\xr,#1)[xnode]{};
        \node (a2#2) at (2*\xr,#1)[xnode]{};
        \node (a3#2) at (3*\xr,#1)[xnode]{};
        \node (ald#2) at (5*\xr,#1)[]{$\cdots$};
        \node (a4#2) at (7*\xr,#1)[xnode]{};
        \node (a5#2) at (8*\xr,#1)[xnode]{};
        \node (a6#2) at (9*\xr,#1)[xnode]{};
      \draw[redge] (a1#2) -- (a2#2) -- (a3#2) -- (ald#2) -- (a4#2) -- (a5#2) -- (a6#2); 
      \draw[redge] (a1#2) to [out=172, in =8](a6#2);
}

      \insta{0}{1}{$V_1$};
      \insta{-0.8}{2}{$V_2$};
\node (ld1) at (2.25*\xr,-1.3)[]{$\vdots$};
      \node (ld2)  at (5*\xr,-1.3)[]{$\vdots$};
      \node (ld3) at (7.75*\xr,-1.3)[]{$\vdots$};
\insta{-2.3}{3}{$V_{k}$};
      
      \draw[eeedge] (a21) to (a22) to (a23) to[out=120,in=240] (a21);

    \end{tikzpicture}
\caption{Illustration for the proof of \cref{lem:wk1-hard}.
      Thick solid edges are the required edges~$R$.
      Thin dashed edges are a colorful cycle
      and,
      at the same time,
        an \EE{}.}
    \label{fig:wk1-hard}
  \end{figure}
$G'=(V,E')$~is a complete graph,
  the set~$R$
  contains a cycle on the vertices in~$V_i$
  for each $i\in\{1,\dots,k\}$,
  and
  \[
    \cost \colon E'\to\N,\, e\mapsto
    \begin{cases}
      1&\text{ if $e\in E\cup R$},\\
      2&\text{ otherwise}.
    \end{cases}
  \]
  Note that,
  since all edge weights~\(\cost\)
  are one and two,
  \(\cost\)~is metric.
  Thus,
  by
  \cref{lem:bounds},
  $|T|+|M|+\cost(T)+\cost(M)\in\Theta(|D|+\cost(D))$.
  Moreover,
  since $\eG{G'}{R}$~is balanced,
  $|T|+|M|+\cost(T)+\cost(M)=|T|+\cost(T)\in O(k)$.
  We show that
  $I$~is a yes\hyp instance
  if and only if
  $I'$ has an \sol{} of weight~$\cost(R)+k=|R|+k$,
  which,
  by \cref{prop:linearlift},
  is equivalent to having
  an \EE{}~$S$ of weight~$\cost(S)\leq k$.
  
  \RD{}
  Let~$S$ be a multicolored cycle in~$G$.
  Since $\eG{G'}{R}$ is a disjoint union of cycles,
  $\eG{G'}{R}$~is balanced.
  Since $S$~is a cycle,
  $\eG{G'}{R\uplus S}$ is also balanced.
  Since $S$~contains one vertex of each color,
  $\eG{G'}{R\uplus S}$~is additionally connected.
  Thus, $S$~is an \EE{} for $(G',R,\cost)$.
  Since $S$~consists of edges of~$G$,
  we conclude \(\cost(S)=|S|=k\).

  \LD{}
  Let $S$~be an \emi{} \EE{} with \(\cost(S)\leq k\)
  for~$(G',\allowbreak R,\cost)$.
  Since $\eG{G'}{R}$ and $\eG{G'}{R\uplus S}$
  are balanced,
  so is $\eG{G'}{S}$.
  Since $\eG{G'}{R\uplus S}$~is connected
  and $\eG{G'}{S}$~is balanced,
  $S$~contains at least two edges incident
  to a vertex in~$V_i$ for each~$i\in\{1,\dots,k\}$.
  Thus,  
  since $\cost(S)\leq k$,
  $\eG{G'}{S}$ has to contain \emph{exactly} $k$~edges,
  all of weight one,
  and \emph{exactly} one vertex
  of~$V_i$ for each~$i\in\{1,\dots,k\}$,
  that is,
  $k$~vertices.
  Since $\eG{G'}{S}$~is balanced,
  it follows that $\eG{G'}{S}$~is a
  collection of cycles
  whose color sets do not intersect.
  Thus,
  if $\eG{G'}{S}$~was not connected,
  then $\eG{G'}{R\uplus S}$ would not be either.
  We conclude that~$\eG{G'}{S}$
  is connected,
  that is,
  a \emph{single} cycle
  containing exactly one vertex of each color.
  By \cref{lem:onebal},
  none of the edges in~$S$ are in~$R$.
  Since all of them have weight one,
  they are in~$G$.
  It follows that $S$~forms
  a multicolored cycle in~$G$.
\end{proof}

\noindent
Having shown WK[1]-hardness
in \cref{lem:wk1-hard},
we now show containment in~WK[1],
concluding the proof of
\cref{thm:wk1}.
Note that we showed hardness
for a parameter in~$\Theta(|D|+\cost(D))$,
whereas containment we show
for an even smaller
parameter in~$O(|D|+\log(1+\cost(D)))\subseteq O(|D|+\cost(D))$.
This means that,
if any problem in WK[1] turns out to have
a polynomial\hyp size Turing kernel,
then there will be a Turing kernel
for RPP with size polynomial even in
$|D|+\log(1+\cost(D))$.

\begin{lemma}
  \label{lem:wk1-containment}
  RPP parameterized by
  $|T|+|M|+\log(1+\cost(T)+\cost(M))\in O(|D|+\log(1+\cost(D)))$
  is in WK[1].
\end{lemma}

\begin{proof}
  We prove a polynomial~parameter transformation
  from RPP parameterized by~$|T|+|M|+\log(1+\cost(T)+\cost(M))$
  to NDTM Halting (\cref{prob:ndtm}).
  By \cref{rem:metric-rpp},
  it is sufficient to reduce RPP
  with triangle inequality.
  To this end,
  we construct a number~$t\in\N$
  and a nondeterministic Turing
  machine~$\M$
that,
  given an empty input string,
  has a computation path
  halting within $t$
  steps if and only if
  a given RPP instance~$I=(G,R,\cost)$
  on a graph~$G=(V,E)$ with $n$~vertices,
  $m$~edges, and
  triangle inequality
  has an \sol{} of weight at most~$\cost(R)+k$,
  that is,
  an \EE{} of weight at most~$k$.
  To this end,
  let
  \begin{align*}
    d_1&:=|M|+2|T|&&\text{and}&d_2&:=\cost(M)+2\cost(T).
  \end{align*}  
By \eqref{M+2T},
  there is an optimal \EE{}
  of at most $d_1$~edges for~$(G,R,\cost)$.
  Thus,
  if $d_1\leq\log n$,
  or $n=0$,
  then we optimally solve $I$~in polynomial time
  \citep{DMNW13}
  and return $t=1$ with
  a Turing machine of constant size
  that immediately halts or never halts
  in dependence of whether $I$~is a
  yes\hyp instance.
  Thus,
  we henceforth assume
  \begin{equation}
    0\leq\log n< d_1.\label{eq:nleqd}
  \end{equation}
  If $k\geq d_2$,
  then, by \eqref{cM+2T}, $I$~is a yes\hyp instance
  and we simply return
  $t=1$ and
  a Turing machine~$\M$
  of constant size
  that immediately halts.
  Thus, we henceforth assume
  \begin{equation}
    0\leq k<d_2.\label{eq:kleqd}
  \end{equation}
  By \eqref{eq:kleqd},
  edges~$e\in E$
  with weight~$\cost(e)\geq d_2$
  will not be part of the sought \EE{}
  of weight~$k$,
  thus we lower their weight to~$d_2$
  and henceforth assume
  \begin{equation}
    \label{eq:eleqd}
    \cost(e)\leq d_2\qquad\text{ for all $e\in E$}.
  \end{equation}
  We now construct a nondeterministic Turing machine~$\M$.
  The Turing machine has a binary alphabet
  and $G$~is assumed to be encoded in binary.
  The state names of~$\M$
  encode the incidence matrix of~$G$,
  the weight~\(\cost(e)\) of each edge~$e\in E$
  in binary,
  and, for each vertex~$v\in V$,
  the number of its connected component in~$\eG{G}{R}$ in binary.
  Turing machine~$\M$ uses three tapes:
  on the \emph{edge tape}, it guesses
  at most $d_1$~edges,
  on the \emph{connection tape},
  it records
  which of the initially $O(d_1)$~connected
  components of~$\eG{G}{R}$
  (by \eqref{TD})
  are connected by the guessed edges,
  and
  on the \emph{balancing tape},
  it records all \imba{} vertices,
  of which initially there are $O(d_1)$
  by \eqref{MD}
  and whose number will never exceed $O(d_1)$
  by adding at most $d_1$~guessed edges.
  The program of Turing machine~$\M$
  is as follows.
  On the empty input,
  at most $d_1$~times:
  \begin{compactenum}
  \item Write the name~$\{u,v\}$ of an arbitrary edge of~$G$
    (listed in the state names)
    onto the edge tape.
    This takes $O(\log n)\in O(d_1)$  steps (by \eqref{eq:nleqd}).
    
  \item Flip the balance of $u$ and~$v$ on
    the balancing tape in
    $\poly(d_1)$ steps
    because there are only $O(d_1)$ vertices on it,
    each of which is encoded in $O(\log n)\subseteq O(d_1)$ bits
    (by \eqref{eq:nleqd}).
    
  \item Record the connectivity
    of the components containing~$u$ and~$v$
    on the connection tape
    in $\poly(d_1)$ steps
    because there are only
    $O(d_1)$~component names on it.
    
\end{compactenum}
  If,
  after at most $d_1$~guessed edges,
  the computation does \emph{not}
  reach a configuration where
  all vertices are \ba{}
  and all components of~$\eG{G}{R}$ are connected,
  then $\M$~goes into an infinite loop.
  Otherwise,
  in $\poly(d_1)$~steps,
  we reached
  such a configuration
  and it remains to check
  whether the guessed edges have weight at most~$k$.
  To this end,
  $\M$~writes down
  the weights of the at most $d_1$~guessed edges
  in binary,
  sums them up,
  and compares them to~$k$
  in~$\poly(d_1+\log d_2)$ steps
  because of \eqref{eq:kleqd} and \eqref{eq:eleqd}.
  If their weight is more than~$k$,
  then $\M$~goes into an infinite loop.
  Otherwise,
  $\M$~stops.
  Observe that
  each computation path of~$\M$,
  if it terminates,
  then it does so within~$t$
  steps for some~$t\in\poly(d_1+\log d_2)$.

  We have shown a correct reduction
  from RPP to NDTM Halting.
  To show that it is a polynomial~parameter transformation,
  it remains to show~\(t\log |\M|\in\poly(d_1+\log d_2)=\poly(|T|+|M|+\log(1+\cost(T)+\cost(M))\).
  Since $t\in\poly(d_1+\log d_2)$,
  it remains to show that
  \(\log|\M|\in\poly(d_1+\log d_2)\).
The graph~$G$
  can be hard\hyp coded
  in
  Turing machine~$\M$ 
  using \(\poly(n)\)~symbols.
  The encoded edge weights
  have total size \(\poly(n+ \log d_2)\)
  by \eqref{eq:eleqd}.
  Its program therefore has size \(|\M|\leq\poly(n+d_1+\log d_2)\).
  Thus,
  \begin{align*}
    \log |\M| &\leq\log(n+d_1+\log d_2)^c=c\log (n+d_1+\log d_2)&\text{for some constant~$c$}\\
              &=c\Bigl(\log n+\log\Bigl(1+\frac{d_1}{n}\Bigr)+\log\Bigl(1+\frac{\log d_2}{n+d_1}\Bigr)\Bigr)&\text{using $\log(a+b)=\log a +\log(1+a/b)$}\\
              &\leq c(d_1+\log(1+d_1)+\log(1+\log d_2))&\text{using \eqref{eq:nleqd},
                                                         $n\geq 1$, and $d_1\geq 1$}\\
    &\leq c(2d_1+\log d_2)\in\poly(d_1+ \log d_2)&\text{using $\log(1+x)\leq x$.}&\qedhere
  \end{align*}
\end{proof}

\section{Approximate kernelization schemes}
\label{sec:psaks}

In \cref{sec:wk1},
we have seen that
provably effective and efficient
data reduction for RPP
is hard when one requires
exact solutions.
In this section,
we show effective data reduction
rules that only slightly decrease
the solution quality.
Indeed,
we will prove \cref{thm:psaks}.
To this end,
in \cref{sec:rednr,sec:redr,sec:ba},
we present three data reduction rules.
In \cref{sec:psaksproof},
we then
show how to apply these rules
to obtain
a polynomial\hyp size approximate kernelization
scheme~(PSAKS)
of size~$2b+O(c/\varepsilon)$,
proving \cref{thm:psaks}.
Finally,
in \cref{sec:psaks-c},
we discuss some problems
that one faces when
trying to improve it
to a PSAKS for RPP parameterized only by~$b$ or only by~$c$.

\subsection{Removing vertices non-incident to required edges}
\label{sec:rednr}
\noindent
Recall that,
by \cref{rem:metric-rpp},
we can assume that the input graph~$G$ is complete
and the edge weights satisfy the triangle inequality.
Thus,
we can simply delete vertices that are not incident to required edges \citep{BNSW15}.

\begin{rrule}
  \label{rrule:delverts}
  Let $(G,R,\cost)$~be an RPP instance
  with triangle inequality.
  Delete all vertices
  that are not incident to edges in~$R$.
\end{rrule}

\noindent
Since,
by \cref{obs:allR},
no \emi{} \EE{} uses vertices
outside of~$V(R)$,
the following proposition is immediate.

\begin{proposition}
  \label{rem:allR}
  \cref{rrule:delverts}
  turns an RPP instance~$(G,R,\cost)$
  into an RPP instance~$(G',R,\cost)$
  such that
  \begin{compactitem}
  \item any \emi{} \EE{} for~$(G,R,\cost)$ is one
    for~$(G',R,\cost)$
    and 
  \item any
    \EE{} for~$(G',R,\cost)$
    is one for~$(G,R,\cost)$.
  \end{compactitem}
\end{proposition}

\subsection{Reducing the number of required edges}
\label{sec:redr}
In this section,
we present a data reduction rule
to shrink the set of required edges.
This will be crucial
since other data reduction rules
only reduce the number of vertices,
yet may leave the multiset of
required edges between them unbounded.

\begin{rrule}
  \label{rrule:delcyc}
  Let $(G,R,\cost)$~be an instance of RPP
  and $C$~be a cycle in~$\eG{G}{R}$
  such that \(\eG{G}{R\setminus C}\)
  has the same number of
  connected components
  as~$\eG{G}{R}$,
  then delete the edges of~\(C\) from~$R$.
\end{rrule}

\begin{lemma}
  \label{lem:delcyc}
  Using \cref{rrule:delcyc},
  one can in $O(|R|)$~time
  compute a set~$R'\subseteq R$
  of required edges
  with the following properties.
  \begin{compactenum}[(i)]
  \item\label{cyc2}
    Any \EE{} for~$(G,R',\cost)$
    is one for~$(G,R,\cost)$
    and vice versa.
  \item\label{cyc1}
    The number of edges
    in each connected component of~$\eG{G}{R'}$
    with $k$~vertices
    is at most $\max\{1,2k-2\}$.
  \end{compactenum}
\end{lemma}
\begin{proof}
  We apply \cref{rrule:delcyc} as follows.
  For $i\in\{1,\dots,c\}$,
  let $R_i\subseteq R$~be the set of required
  edges in the $i$-th connected component of~$\eG{G}{R}$.
  In $O(|R_i|)$~time,
  one can compute
  a depth\hyp first search tree~$T_i$ of~$\eG{G}{R_i}$,
  which is a spanning tree of~$\eG{G}{R_i}$.
  Now we remove all cycles from~$\eG{G}{R_i\setminus T_i}$ as follows.
  We start a depth\hyp first search on~$\eG{G}{R_i\setminus T_i}$.
  Whenever we meet a vertex~$v$ a second time,
  we backtrack to the previous occurrence of~$v$,
  deleting all visited edges from the graph on the way.
  This procedure removes
  all cycles from~$\eG{G}{R_i\setminus T_i}$
  and looks at each edge in~$R_i\setminus T_i$
  at most twice,
  thus works
  in $O(|R_i|)$~time.
  
  \eqref{cyc2}
  Any two vertices
  are connected in~$\eG{G}{R}$
  if and only if they
  are connected in~$\eG{G}{R'}$.
  Moreover,
  the balance
  of each vertex
  is the same in~$\eG{G}{R}$ and~$\eG{G}{R'}$.

  \eqref{cyc1}
  Each component of~$\eG{G}{R'}$
  with $k=1$~vertex has one edge (a loop).
  Each component of~$\eG{G}{R'}$
  with $k>1$~vertices
  consists of \(k-1\)~edges of a spanning tree~$T_i$
  for some~$i\in\{1,\dots,c\}$
  and at most \(k-1\)~additional edges,
  otherwise they would contain a cycle.
\end{proof}

\subsection{Reducing the number of balanced vertices}
\label{sec:ba}

In this section,
we present a data reduction rule
that removes balanced vertices.
To this end,
we introduce an operation
that allows us to remove balanced vertices
while maintaining the balance
of their neighbors.

First,
the following lemma in particular shows
that removing a balanced
vertex with all its incident edges
changes the balance of
an even number of vertices.
This allows us to restore
their original balance
by adding a matching to the set
of required edges,
not increasing
the total weight of required edges.
This will be crucial to prove
that our reduction rules
maintain approximation factors.

\begin{lemma}
  \label{lem:sew}
  Let $\Gamma=(V,E)$~be a graph,
  \(\cost\colon \{\{u,v\}\mid u,v\in V\}\to\N\)
  satisfy the triangle inequality,
  and
  $F$~be an even\hyp cardinality submultiset
  of edges incident to a common vertex~$v\in V$.
  Then
  \begin{compactenum}[(i)]
  \item\label{sew1}
    The set~$U\subseteq V\setminus\{v\}$
    of vertices
    incident to an odd number of edges of~$F$
    has even cardinality.
    
  \item\label{sew2} For any
  matching~$M_v$ in the complete graph on~$U$,
  \(\cost(M_v)\leq\cost(F)\) and $|M_v|\leq|F|$.
  \end{compactenum}
\end{lemma}

\begin{proof}
  \eqref{sew1}
  Any graph, in particular~$\eG{\Gamma}{F}$,
  has an even number of odd\hyp degree vertices.
  Since $|F|$~is even,
  $v$~is not one of them.

  \eqref{sew2}
  Let $e_i:=\{x_i,y_i\}$
  for~$i\in\{1,\dots,|M_v|\}$
  be the edges of~$M_v$.
  Then there are
  pairwise edge\hyp disjoint
  paths~$p_i:=(x_i,v,y_i)$
  for~$i\in\{1,\dots,|M_v|\}$
  in~$\eG{\Gamma}{F}$.
  Thus
  \[
    \cost(M_v)=\sum_{i=1}^{|M_v|}\cost(e_i)\leq \sum_{i=1}^{|M_v|}\cost(p_i)\leq \cost(F).
  \]
  Finally,
  $|M_v|\leq |U|\leq|F|$.\end{proof}

\noindent
We now use \cref{lem:sew}
to define an operation that allows
us to remove a balanced vertex from~$\eG{G}{R}$.
It is illustrated in
\cref{fig:extr}.

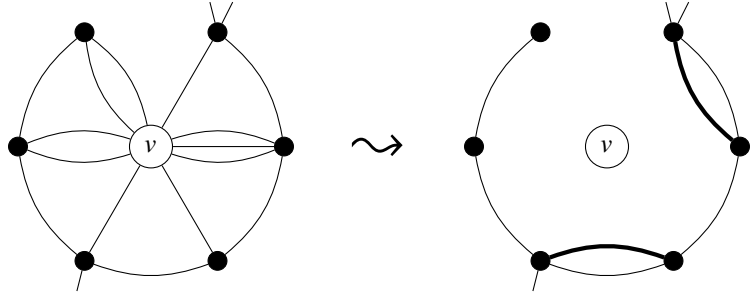
\begin{figure}[t]
  \centering
\begin{tikzpicture}
\def\xr{1}
	\def\yr{1}
	\def\dlea{1.75}
	\tikzstyle{xnode}+=[circle, fill, scale=3/4, draw]
	\tikzstyle{xedge}+=[-]

	\begin{scope}
		\node (v) at (0,0)[circle,draw]{$v$};
		\foreach \j in {1,...,6}{
				   \node (v\j) at ($(v)+ (\j*360/6:\dlea cm)$)[xnode]{};
	 		 }
		\draw[xedge] (v) to (v1);
		\draw[xedge] (v) to [bend left=20](v2);
		\draw[xedge] (v) to [bend right=20](v2);
		\draw[xedge] (v) to [bend left=20](v3);
		\draw[xedge] (v) to [bend right=20](v3);
		\draw[xedge] (v) to (v4);
		\draw[xedge] (v) to (v5);
		\draw[xedge] (v) to [bend left=20](v6);
		\draw[xedge] (v) to (v6);
		\draw[xedge] (v) to [bend right=20](v6);
	
		\draw[xedge] (v1) to [bend left=20](v6);
		\foreach \j in {6,...,3} { \pgfmathtruncatemacro{\i}{\j -1}; \draw[xedge] (v\j) to [bend left=20](v\i); }
				
		\draw[xedge] (v1) to ($(v1)+(-0.1*\xr,0.4*\yr)$);		
		\draw[xedge] (v1) to ($(v1)+(0.2*\xr,0.4*\yr)$);
		\draw[xedge] (v4) to ($(v4)+(-0.1*\xr,-0.4*\yr)$);		
	\end{scope}

	\node at (3*\xr,0)[scale=2]{$\leadsto$};

	\begin{scope}[xshift=6*\xr cm]
		\node (v) at (0,0)[circle,draw]{$v$};
		\foreach \j in {1,...,6}{
				   \node (v\j) at ($(v)+ (\j*360/6:\dlea cm)$)[xnode]{};
	 		 }
	
		\draw[xedge] (v1) to [bend left=20](v6);
		\draw[xedge,ultra thick] (v1) to [bend right=20](v6);
		\foreach \j in {6,...,3} { \pgfmathtruncatemacro{\i}{\j -1}; \draw[xedge] (v\j) to [bend left=20](v\i); }
		\draw[xedge,ultra thick] (v5) to [bend right=20](v4);

		\draw[xedge] (v1) to ($(v1)+(-0.1*\xr,0.4*\yr)$);		
		\draw[xedge] (v1) to ($(v1)+(0.2*\xr,0.4*\yr)$);
		\draw[xedge] (v4) to ($(v4)+(-0.1*\xr,-0.4*\yr)$);		
	\end{scope}
	
\end{tikzpicture}
\caption{Illustration of \cref{def:extr}\eqref{extr2}.
  Only required edges are shown.
  Thick edges on the right
  are the added matching~$M_v$.}
  \label{fig:extr}
\end{figure}

\begin{definition}[vertex extraction]
  \label{def:extr}
  Let $(G,R,\cost)$~be an RPP instance
  with $\cost$~satisfying the triangle inequality,
  $v$~be a vertex that
  \begin{compactitem}
  \item is balanced in a connected component of~$\eG{G}{R}$
    with at least three vertices and
  \item not a cut vertex of~$\eG{G}{R}$ or contained in exactly two blocks of~$\eG{G}{R}$,
  \end{compactitem}
  and let $R_v\subseteq R$~be
  the required edges incident to~$v$.
  The result of \emph{extracting~$v$ from~$\eG{G}{R}$} is defined as follows.
  \begin{compactenum}[(a)]
\item
    \label{extr2}
    If $v$~is not a cut vertex of~$\eG{G}{R}$,
    then let $M_v$~be any perfect matching on
    the set of vertices
    incident to an odd number of edges of~$R_v$.
    The result of extracting~$v$ is $R'=(R\setminus R_v)\uplus M_v$.
  \item
    \label{extr3}
    If $v$~is a cut vertex of~$\eG{G}{R}$
    contained in exactly two blocks~$A$
    and~$B$ of~$\eG{G}{R}$,
    then let $a$~be a neighbor of~$v$ in~$A$,
    $b$~be a neighbor of~$v$ in~$B$,
    and $R'=(R\setminus\{\{a,v\},\{b,v\}\})\uplus\{a,b\}$.
    \begin{compactenum}[1.]
    \item\label{extr3a}     If $v$~is not contained in~$\eG{G}{R'}$,
    then $R'$ is the result of extracting~$v$.
  \item\label{extr3b}     Otherwise,
    $v$~is not a cut vertex of~$\eG{G}{R'}$
    and the result of extracting~$v$ from~$\eG{G}{R}$
    is defined as the result of extracting~$v$ from~$\eG{G}{R'}$.
  \end{compactenum}
  \end{compactenum}
\end{definition}

\begin{lemma}
  \label{lem:extrcorrect}
  Let $(G,R,\cost)$~be
  an RPP instance
and $R'$~be the result
  of extracting a balanced vertex~$v$ of~$\eG{G}{R}$.
  Then the following properties hold.
  \begin{compactenum}[(i)]
  \item
    \label{extrcorr1b}
    \(V(R')= V(R)\setminus\{v\}\).
  \item
    \label{extrcorr1}
    \(\cost(R')\leq\cost(R)\) and $|R'|\leq|R|$.
  \item
    \label{extrcorr2}
    Each vertex of~$\eG{G}{R'}$
    is balanced if and only
    if it is balanced in~$\eG{G}{R}$.
  \item
    \label{extrcorr3}
    Two vertices of~$\eG{G}{R'}$
    are connected if and only if
    they are so in~$\eG{G}{R}$.
  \item
    \label{extrcorr1c}
    Any multiset~$S$
    of edges with~$V(S)\subseteq V(R')$
    is an \EE{} for~\((G,R',\cost)\)
    if and only if it is one for~$(G,R,\cost)$.
  \end{compactenum}
\end{lemma}

\begin{proof}
  \eqref{extrcorr1b}
First, assume that
  $R'$~was obtained
  according to  \cref{def:extr}\eqref{extr2}.
  Let $R_v\subseteq R$~be
  the required edges incident to~$v$
  and
  $U\subseteq V$~be
  the set of vertices incident to
  edges of~$R_v$.
Obviously,
  $V(R')\subseteq V(R)\setminus\{v\}$
  and
  $V(R')\supseteq V(R)\setminus U$.
  Moreover,
  $V(R')\supseteq U\setminus\{v\}$:
  since $v$~is in a connected component
  of~$\eG{G}{R}$ with at least three vertices
  but not a cut vertex,
  the vertices in~$U\setminus\{v\}$
  are incident to edges in~$R\setminus R_v$.
  These are retained in~$\eG{G}{R'}$.
  If $R'$~was obtained
  according to \cref{def:extr}\eqref{extr3},
  then $R'$~is the same as
  if it were obtained from
  \cref{def:extr}\eqref{extr2}
  by extracting~$v$ from~$\eG{G}{R\uplus\{\{a,b\}\}}$,
  where it is not a cut vertex.

\eqref{extrcorr1}--\eqref{extrcorr3}
  If $R'$~was obtained
  according to
  \cref{def:extr}\eqref{extr2},
  then
  \eqref{extrcorr1} and \eqref{extrcorr2}
  follow from
  \cref{lem:sew}
  applied to~$\Gamma=\eG{G}{R}$
  and~$F=R_v$,
  whereas \eqref{extrcorr3} is clear
  since $v$~is not a cut vertex of~$\eG{G}{R}$
  and $v$~is not in~$\eG{G}{R'}$.
  If $R'$~was obtained
  according to  \cref{def:extr}\eqref{extr3a},
  then \eqref{extrcorr1}--\eqref{extrcorr3}
  trivially hold since $\cost(\{a,b\})\leq\cost(\{a,v\})+\cost(\{b,v\})$
  and $v$~is not in~$\eG{G}{R'}$.
  If $R'$~was obtained
  according to  \cref{def:extr}\eqref{extr3b},
  then \eqref{extrcorr1}--\eqref{extrcorr3}
  hold since~$R'$
  is the same as extracting~$v$
  from~$\eG{G}{(R\setminus\{\{a,v\},\{b,v\}\})\uplus\{\{a,b\}\}}$,
  where it is not a cut vertex,
  and from
  $\cost(\{a,b\})\leq\cost(\{a,v\})+\cost(\{b,v\})$.

  \looseness=-1
  \eqref{extrcorr1c}
  We show that $\eG{G}{R\uplus S}$
  is connected and balanced
  if and only if~$\eG{G}{R'\uplus S}$ is.

  \looseness=-1
  \textit{Connectivity.}
  By \eqref{extrcorr3},
  two vertices of~$V(R')$
  are connected in~$\eG{G}{R'}$
  if and only if they are connected in~$\eG{G}{R}$.
  Since $V(S)\subseteq V(R')\subseteq V(R)$
  by \eqref{extrcorr1b},
  two vertices in~$V(R')=V(R'\uplus S)$
  are connected
  in~$\eG{G}{R'\uplus S}$
  if and only if
  they are connected in~$\eG{G}{R\uplus S}$.
  By \eqref{extrcorr1b},
  the only vertex of~$\eG{G}{R\uplus S}$
  that is absent from~$\eG{G}{R'\uplus S}$ is~$v$,
  which is not isolated in~$\eG{G}{R\uplus S}$
  since it is not isolated in~$\eG{G}{R}$.
  
  \textit{Balance.}
  By \eqref{extrcorr2},
  each vertex in~$V(R')$ is
  balanced in~$\eG{G}{R'}$ if and only if it is balanced in~$\eG{G}{R}$.
  Since $V(S)\subseteq V(R')\subseteq V(R)$
  by \eqref{extrcorr1b},
  each vertex in~$V(R')=V(R'\uplus S)$
  is balanced in~$\eG{G}{R'\uplus S}$
  if and only if it is balanced in~$\eG{G}{R\uplus S}$.
  By \eqref{extrcorr1b},
  the only vertex in~$\eG{G}{R\uplus S}$
  that is absent from~$\eG{G}{R'\uplus S}$
  is~$v$.
  If so,
  then $v\notin V(S)$
  and $v$~is balanced in~$\eG{G}{R\uplus S}$
  because it is balanced in~$\eG{G}{R}$.
\end{proof}

\noindent
\looseness=-1
We can now turn \cref{def:extr} into
a data reduction rule.
Its parameter~\(\pBal{}\in\Q\)
allows a trade\hyp off
between its strength
and the introduced error.

\begin{rrule}
  \label{rrule:delpath}
  \label{rrule:ballred}
  \looseness=-1
  Let $(G,R,\cost)$~be an RPP instance
  with~$G=(V,E)$,
  \(\cost\)~satisfying the triangle inequality,
  and \(\pBal{}\in\Q\).
  Let $C_i$~be the vertices
  in connected component~$i\in\{1,\dots,c\}$ of~$\eG{G}{R}$
  and $B_i\subseteq C_i$~be an
  inclusion\hyp maximal set of vertices
  such that,
  for each~$u,v\in B_i$
  with $u\ne v$, one has \(\cost(\{u,v\})>\pBal{}\).
  Finally,
  let
  \[
    B:=\bigcup_{i=1}^c B_i.
  \]
  Now,
  initially let~$R':=R$
  and,
  as long as $\eG{G}{R'}$~contains
  a vertex~$v\in V\setminus B$
  that can be extracted using \cref{def:extr},
  replace $R'$~by the result of extracting~$v$.
\end{rrule}
\noindent
We now analyze the effectivity and error of
\cref{rrule:ballred}.

\begin{lemma}
  \label{lem:ballred}
  Let $(G,R,\cost)$ be an RPP instance
  with \(\cost\)~satisfying the triangle inequality.
Then,
  \cref{rrule:delpath}
  in $O(n^3)$~time
  yields a multiset~$R'$ of edges such that
  \begin{compactenum}[(i)]
  \item\label{ballred1} $\cost(R')\leq\cost(R)$ and $V(R')\subseteq V(R)$.
  \item\label{ballred2}
    Any multiset~$S$ of edges
    with $V(S)\subseteq V(R')$
    is an \EE{}
    for~\((G,R',\cost)\)
    if and only if it is one for~\((G,R,\cost)\).
  \item\label{ballred3}
    Any \emi{} \EE{}~$S$
    for~\((G,R,\cost)\)
    can be turned into an \EE{}~$S'$
    for~\((G,R',\cost)\)
    such that \(\cost(S')\leq \cost(S)+2\pBal{}\cdot (2c-2)\).
    
  \item\label{ballred4}
    $\eG{G}{R'}$ contains at most
    \(2b+2c+4\cost(R)/\pBal{}\)~vertices.
  \end{compactenum}
\end{lemma}
\begin{proof}
  \eqref{ballred1} and \eqref{ballred2}
  follow from \cref{lem:extrcorrect}
  since $R'$~is the result
  of a sequence of vertex extractions.

  \eqref{ballred3}
  We turn~$S$
  into an \EE{}~$S'$
  with~$V(S')\subseteq V(R')$
  and then apply~\eqref{ballred2}.
  First,
  since $S$~is \emi{}
  and \(\cost\)~satisfies the triangle inequality,
  by \cref{obs:allR},
  $V(S)\subseteq V(R)$.
  By \cref{rrule:ballred},
  the vertices in~$X:=V(R)\setminus V(R')$
  are not in~$B$
  and, thus,
  for each~$v\in X\cap C_i$,
  we find a vertex~$v'\in B_i$
  such that~$\cost(\{v,v'\})\leq\pBal{}$.
  Note that~$v'\in V(R')$.
  Since each vertex in~$X$
  is balanced in~$\eG{G}{R}$,
  by \cref{rem:one-or-two},
  each vertex~$v\in X\cap V(S)$
  is incident to exactly \emph{two}
  edges~$\{v,u\}$ and~$\{v,w\}$ of~$S$
  (possibly, $u=w$).
  Since $\{v,v'\}\subseteq C_i$,
  $S':=(S\setminus\{\{v,u\},\{v,w\}\})\uplus \{v',u\}\uplus\{v',w\}$
  is also an \EE{} for~$(G,R,\cost)$.
  Moreover,
  \(\cost(S')\leq\cost(S)+2\pBal{}\).
  Doing this replacement
  for each~$v\in X\cap V(S)$,
  we finally obtain an \EE{}~$S'$
  for~$(G,R,\cost)$
  with $V(S')\subseteq V(R')$
  and \(\cost(S')\leq\cost(S)+2\pBal{}\cdot|X\cap V(S)|\).
  Since each vertex in~$X$ is balanced in~$\eG{G}{R}$,
  by \cref{lem:2C-2},
  $|X\cap V(S)|\leq 2c-2$.
  Finally,
  by \eqref{ballred2},
  $S'$~is an \EE{} for~$(G,R',\cost)$.

  \eqref{ballred4}
  The vertices of~$\eG{G}{R'}$ can be partitioned
  into~$X\uplus Y\uplus Z$,
  where $X$~are imbalanced in~$\eG{G}{R'}$,
  $Y$~are balanced and in~$B$,
  and $Z$~are balanced but not in~$B$.
  
  By \cref{lem:extrcorrect}\eqref{extrcorr2},
  the vertices in~$X$ are imbalanced in~$\eG{G}{R}$ also.
  Thus,
  \begin{equation}
    |X|\leq b.\label{xleqb}
  \end{equation}
  We next analyze~$|Y|$.
  For $i\in \{1,\dots,c\}$,
  let~$R_i\subseteq R$ be
  the edges between vertices in~$C_i$,
  $T_i^*$~be the edge set of a
  tree of least weight in~$\eG{G}{R_i}$
  connecting all vertices in~$B_i$,
  $T_i$~be the edge set of
  a minimum\hyp weight spanning tree in~$G[B_i]$,
  and
  $H_i$~be the edge set of a
  minimum\hyp weight Hamiltonian cycle in~$G[B_i]$.
  Doubling all edges of~$T_i^*$
  yields a closed walk
  in~$\eG{G}{R_i}$ containing the vertices in~$B_i$.
  Using the triangle inequality of~$\cost$,
  it can be shortcut to a Hamiltonian cycle in~$G[B_i]$.
  Thus,
  $\cost(T_i)\leq\cost(H_i)\leq 2\cost(T_i^*)$.\footnote{That is,
    $T_i$~is the folklore 2-approximation
    of a Steiner tree with terminals~$B_i$ in~$\eG{G}{R_i}$.}
  We thus get
  \begin{align}
    (|B_i|-1)\pBal{}&=\sum_{e\in T_i}\pBal{}<\sum_{e\in T_i}\cost(e)=
                     \cost(T_i)\leq2\cost(T_i^*)\leq 2\cost(R_i)\notag{}\text{\quad and thus}\\
    |Y|\leq|B|&=\sum_{i=1}^c|B_i|<\sum_{i=1}^c\Biggl(\frac{2\cost(R_i)}\pBal{}+1\Biggr)=2\omega(R)/\pBal{}+c.
                \label{yleqr}
  \end{align}
  \looseness=-1
  Finally,
  we analyze~$|Z|$.
  \cref{def:extr} is not applicable to any vertex~$v\in Z$,
  since it would have been removed by \cref{rrule:ballred}.
  Thus,
  $v$~is a cut vertex contained
  in at least three blocks of~$\eG{G}{R'}$
  or its connected component of~$\eG{G}{R'}$
  consists of only two vertices.
  To analyze~$|Z|$,
  for each~$i\in\{1,\dots,c\}$,
  consider $X_i:=X\cap C_i$,
  $Z_i:=Z\cap C_i$,
  the set $R_i'\subseteq R'$ of edges
  between vertices in~$C_i$,
  and the \emph{block-cut tree}~$T_i$ of~$\eG{G}{R_i'}$:
  the vertices of~$T_i$
  are the cut vertices
  and the blocks of~$\eG{G}{R_i'}$
  and there is an edge between a cut vertex~$v$
  and a block~$A$ of~$\eG{G}{R_i'}$
  in~$T_i$ if $v$~is contained in~$A$.
  Then either $|Z_i|\leq 2$ or
  the vertices in~$Z_i$
  have degree at least three in~$T_i$.
Therefore,
  $T_i$~has at most $|X_i|+|Y_i|$~leaves.
  Since a tree has at least two leaves,
  we get $|X_i|+|Y_i|\geq 2$.
  Moreover,
  since a tree with $\ell$~leaves
  has at most $\ell-1$~vertices
  of degree three,
  $|Z_i|\leq\max\{2,|X_i|+|Y_i|-1\}\leq |X_i|+|Y_i|$.
  Thus,
  \begin{align}
    |Z|=\sum_{i=1}^c|Z_i|\leq |X|+|Y|.
  \label{zleqxy}
  \end{align}
  Combining \eqref{xleqb}, \eqref{yleqr},
  \eqref{zleqxy},
  and that $|V(R')|=|X|+|Y|+|Z|$,
  \eqref{ballred4} follows. 
    
  We finally analyze the running time of \cref{rrule:ballred}.
  For \(i\in\{1,\dots,c\}\),
  all sets $C_i$ and~$B_i$
  can be computed in~$O(n^2)$~time.
  Also
  the blocks of~$\eG{G}{R'}$
  required by \cref{def:extr} are computable
  in $O(n^2)$~time using depth\hyp first search.
  Thus,
  in $O(n)$~time,
  we can find a vertex~$v$
  to which \cref{def:extr} applies.
  Vertex~$v$ can then be extracted in $O(n)$~time
  since the matchings~$M_v$
  in \cref{def:extr} can be chosen arbitrarily,
  that is,
  in particular greedily in $O(n)$~time,
  and the blocks can be recomputed in $O(n^2)$~time.
Finally,
  we extract at most $n$~vertices.
\end{proof}

\subsection{\boldmath A polynomial\hyp size approximate kernelization scheme for the parameter~$b+c$ (proof of \cref{thm:psaks})}
\label{sec:psaksproof}

This section proves \cref{thm:psaks}.
We describe how to transform
a given RPP instance~$I$ and~$\varepsilon>0$
into
an RPP instance~$I'$
such that any $\alpha$\hyp approximate solution
for~$I'$ can be transformed into
an $\alpha(1+\varepsilon)$\hyp approximate solution for~$I$.
Due to \cref{lem:triangle},
we assume that~$I=(G,R,\cost)$
has been preprocessed in $O(n^3)$~time
so as to satisfy the triangle inequality.

\subsubsection{Shrinking the graph}
\label{step:shrinkgraph}
Choose~$\varepsilon_1+\varepsilon_2=\varepsilon$.
Apply \cref{rrule:ballred} with
\begin{align}
  \pBal{}=\frac{\varepsilon_1\cdot \cost(R)}{4c-4},
  \label{gamma}
\end{align}
which, by \cref{lem:ballred},
in $O(n^3)$~time
gives an instance~$(G,R_1,\cost)$
with
\begin{align}
  |V(R_1)|\leq 2b+2c+\frac{16c-16}{\varepsilon_1}.
  \label{vertbound}
\end{align}
To $(G,R_1,\cost)$,
we apply \cref{rrule:delcyc},
which, by \cref{lem:delcyc},
in $O(|R|)$~time
gives an instance~$(G,R_2,\cost)$ with
\begin{align}
  R_2\subseteq R_1\text{\qquad and\qquad}
  |R_2|\leq 4b+4c+\frac{32c-32}{\varepsilon_1}.
  \label{edgebound}
\end{align}
Finally,
applying \cref{rrule:delverts} to~$(G,R_2,\cost)$
in linear time
yields an instance~$(G_2,R_2,\cost)$
such that
\begin{align}
  |V(G_2)|\leq |V(R_2)|\leq|V(R_1)|.
  \label{gvertbound}
\end{align}

\subsubsection{Shrinking edge weights}
\label{step:weights}
Since~$\eG{G}{R\uplus T}$ is connected,
due to the triangle inequality of~$\cost$,
each edge~$e=\{u,v\}$ of~$G$,
and thus of its subgraph~$G_2$,
satisfies \(\cost(e)\leq\cost(R)+\cost(T)\).
Moreover,
by \cref{lem:bounds},
any \emi{} \EE{} for~$(G_2,R_2,\cost)$
has at most~$|M|+2|T|=b/2+2c-2$~edges.
Thus,
we can apply \cref{lem:looseweight}
with
\(\pWR=\cost(R)+\cost(T)\)
and
$N=|R_2|+b/2+2c-2$
to~$(G_2,R_2,\cost)$
to get an instance~$(G_2,R_2,\cost_2)$
such that for all edges~$e$,
\begin{align}
  \cost(e)\leq \frac{|R_2|+b/2+2c-2}{\varepsilon_2}\in O((b+c)/(\varepsilon_1\varepsilon_2))
  \label{weightbound}.
\end{align}
In \cref{lem:looseweight},
set~$\mathcal F$ just contains all vectors~$x$
that encode \sol{}s~$W$
induced by \emi{} \EE{}s for~$(G_2,R_2,\cost)$
(its entries describe how often each edge is in~$W$).
We finally return~$(G_2,R_2,\cost_2)$,
whose construction takes $O(n^3+|R|)$~time,
as required by \cref{thm:psaks}.

\subsubsection{Kernel size analysis}
The returned instance satisfies
\cref{thm:psaks}\eqref{psaks1}
due to \eqref{vertbound} and \eqref{gvertbound},
\eqref{psaks2} due to \eqref{edgebound},
and \eqref{psaks3} due to \eqref{weightbound}.

\subsubsection{Approximation factor analysis}
\label{sec:factor}
It remains to prove \cref{thm:psaks}\eqref{psaks4},
that is,
that we can lift an \(\alpha\)\hyp approximate
solution for~$(G_2,R_2,\cost_2)$
to an $\alpha(1+\varepsilon)$\hyp approximate solution
for~$(G,R,\cost)$.

An optimal \sol{} for~$(G,R,\cost)$
has weight~$\cost(W^*)=\cost(R)+\cost(D)$
by \eqref{cW=R+D},
where $D$~is a minimum\hyp weight \EE{}.
By \cref{lem:ballred}\eqref{ballred3}
and \eqref{gamma},
there is an \EE{}~$D'$ for~$(G,R_1,\cost)$
with
\begin{align}
  \cost(D')\leq \cost(D)+2\pBal(2c-2)=\cost(D)+\varepsilon_1\cdot \cost(R)
  \label{dstrich}.
\end{align}
By \cref{lem:delcyc},
$D'$~is an \EE{} for~$(G,R_2,\cost)$
and, by \cref{rem:allR}, for~$(G_2,R_2,\cost)$.
Then $D'$~is also an \EE{} for~$(G_2,R_2,\cost_2)$.
Thus,
an optimal \sol{} for~$(G_2,R_2,\cost_2)$
has weight at most
$\cost_2(R_2)+\cost_2(D')$.
By \cref{prop:linearlift},
an $\alpha$-approximate solution for~$(G_2,R_2,\cost_2)$,
can be turned into an \EE{}~$S$
such that
\begin{align}
  \cost_2(R_2)+\cost_2(S)\leq\alpha(\cost_2(R_2)+\cost_2(D')).
  \label{firstcost}
\end{align}
By \cref{rem:allR},
$S$~is an \EE{} for~$(G,R_2,\cost)$.
By \cref{lem:delcyc},
$S$~is an \EE{} for~$(G,R_1,\cost)$,
and by \cref{lem:ballred},
it is one for~$(G,R,\cost)$,
since $V(S)\subseteq V(G_2)=V(R_2)\subseteq V(R_1)\subseteq V(R)$.
Thus,
by \cref{prop:linearlift},
$S$~can be turned into an \sol{}
of weight~$\cost(R)+\cost(S)$ for~$(G,R,\cost)$.
We analyze this weight.
By \eqref{firstcost} and
\cref{lem:looseweight} with $\pWR=\cost(R)+\cost(T)$,
\begin{align*}
  \cost(R_2)+\cost(S)&\leq\alpha(\cost(R_2)+\cost(D'))+\varepsilon_2(\cost(R)+\cost(T)).\\
  \intertext{Using $\cost(R_2)\leq\cost(R_1)\leq\cost(R)$
  from \cref{lem:ballred,lem:delcyc},
  and \(\alpha\geq 1\),  we get}
  \cost(R)+\cost(S)&\leq\alpha(\cost(R)+\cost(D'))+\varepsilon_2(\cost(R)+\cost(T))\\
                     &\leq\alpha(\cost(R)+\cost(D'))+\varepsilon_2(\cost(R)+\cost(D))&\text{using \eqref{TD}\phantom.}\\
  &\leq \alpha(\cost(R)+\cost(D)+\varepsilon_1\cost(R))+\varepsilon_2(\cost(R)+\cost(D))&\text{using \eqref{dstrich}\phantom.}\\
  &\leq\alpha(1+\varepsilon_1+\varepsilon_2)(\cost(R)+\cost(D))=\alpha(1+\varepsilon)\cost(W^*)&\text{using \eqref{cW=R+D}.}
\end{align*}
Thus,
we got an \(\alpha(1+\varepsilon)\)\hyp approximation
for~$(G,R,c)$.

\subsection{\boldmath On polynomial\hyp size approximate
  kernelization schemes for the parameters~$b$ and $c$}
\label{sec:psaks-c}
In the previous section,
we have shown
a polynomial\hyp size approximate kernelization
scheme (PSAKS)
for RPP parameterized by~$b+c$.
An obvious question is whether there is a PSAKS
for the parameters~$b$ or~$c$.
For the parameter~$b$,
we can easily answer this question.

\begin{proposition}
  If RPP parameterized by~$b$ has a $(1+\varepsilon)$\hyp approximate
  kernel for any $\varepsilon<1/122$, then P${}={}$NP.
\end{proposition}

\begin{proof}
  Assume that RPP has an $(1+\varepsilon)$\hyp approximate kernel
  of any size~$g(b)$
  for $\varepsilon<1/122$.
  We show how to find an $(1+\varepsilon)$\hyp approximate solution
  for the metric Traveling Salesman problem
  in polynomial time,
  which implies P${}={}$NP \citep{KLS15}.
  Given an instance~$I$ of the metric Traveling Salesman problem,
  create an instance~$I'$ of RPP by adding a required zero\hyp weight loop
  to each vertex.
  Compute an $(1+\varepsilon)$\hyp approximate kernel~$I''$
  for~$I'$ in polynomial time.
  Since $I'$~has no imbalanced vertices,
  the kernel $I''$ has size $g(b)=g(0)\in O(1)$.
  Computing an optimal solution in~$I''$ thus takes constant time,
  can be lifted to an $(1+\varepsilon)$\hyp approximate solution
  for~$I'$ in polynomial time and,
  after removing the loops,
  is an $(1+\varepsilon)$\hyp approximation for~$I$.
\end{proof}

\noindent
Answering the question about the existence of a PSAKS for the parameter~$c$ is not quite as simple.
In the following,
we discuss the difficulties in resolving this question
and
make some first steps towards its resolution.
In particular, we will show a PSAKS for the parameter~$\cost(T)$.

To get the PSAKS for~$c$,
one has to reduce the number of imbalanced vertices in~$\eG{G}{R}$.
An obvious idea to do so
is adding to~$R$
cheap edges of a minimum\hyp weight
perfect matching~$M$ on imbalanced vertices,
since this is optimal
if it happens to connect $\eG{G}{R}$.

\begin{rrule}
  \label{rrule:matching}
  Let $(G,R,\cost)$~be an RPP instance
  with triangle inequality and~\(\pM\in\Q\).
  Add to~$R$ a subset~$M^*\subseteq M$ of edges with
  \[
    \sum_{e\in M^*}\cost(e)\leq\pM.
  \]
\end{rrule}

\begin{observation}
  \label{obs:mstar}
  Let $R'=R\uplus M^*$~be obtained by applying
  \cref{rrule:matching} to~$R$.
  \begin{compactenum}[(i)]
  \item\label{match0}
    There are at most $2(|M|-|M^*|)$
    \imba{} vertices in~$\eG{G}{R'}$.

  \item\label{match1}
    For any \EE{}~$S'$ for~$(G,R',\cost)$,
    \(S=S'\uplus M^*\)~is an \EE{} for~$(G,R,\cost)$
    and \(\cost(R)+\cost(S)=\cost(R')+\cost(S')\).
  
  \item\label{match2}
    For any \EE{}~$S$ for~$(G,R,\cost)$,
    $S'=S\uplus M^*$~is an \EE{} for $(G,R',\cost)$
    with \(\cost(S')\leq \cost(S)+\pM\).
  \end{compactenum}
\end{observation}
\noindent
To show a PSAKS with respect to the parameter~$c$,
this reduction rule is unsuitable for two reasons:
\begin{compactenum}
\item
To reduce the number of
  imbalanced vertices in~$\eG{G}{R}$
  to some constant,
  we have to add all but a constant number of edges of~$M$ to~$R$,
  yet, by \cref{obs:mstar}\eqref{match2},
  each added edge potentially contributes to the error
  and thus would merely retain a 2\hyp approximation.
  Unfortunately,
  \cref{fig:matchtight} shows that
  the bound given by \cref{obs:mstar}\eqref{match2} is tight.
  \begin{figure}[t]
    \centering
\begin{tikzpicture}
\tikzstyle{xnode}=[fill,circle,scale=3/4,draw]
      \tikzstyle{lnode}=[text width=4.5*\xr cm, anchor=west, scale=1]
      \tikzstyle{redge}=[] \tikzstyle{medge}=[ultra thick]
      \tikzstyle{oedge}=[densely dashed]

      \def\xr{1}
      \def\yr{1}

      \newcommand{\bstar}[3]{
        \def\noleaves{8};
        \def\distleaves{1.5*\xr};
        \node (#1) at (#2,#3)[xnode]{};
\foreach \j in {1,...,\noleaves}{
              \node (a\j) at ($(#1)+ (\j*320/\noleaves:\distleaves cm)$)[xnode]{};
            \draw[redge] (#1) -- (a\j);
            }
      }

      \begin{scope}
      \bstar{a}{0}{0};

      \node (b) at (3*\xr,1*\yr)[xnode]{};
      \draw[-,redge] (b) to [out=45,in=90]($(b)+(1*\xr,0)$);
      \draw[-,redge] (b) to [out=-45,in=-90]($(b)+(1*\xr,0)$);
      \node (c) at (3*\xr,0*\yr)[xnode]{};
      \draw[-,redge] (c) to [out=45,in=90]($(c)+(1*\xr,0)$);
      \draw[-,redge] (c) to [out=-45,in=-90]($(c)+(1*\xr,0)$);
      \node (d) at (3*\xr,-1*\yr)[xnode]{};
      \draw[-,redge] (d) to [out=45,in=90]($(d)+(1*\xr,0)$);
      \draw[-,redge] (d) to [out=-45,in=-90]($(d)+(1*\xr,0)$);

      \draw[-,oedge] (a1) -- (b) -- (c) -- (d) -- (a8);
      \foreach \j in {1,3,5,7}{\pgfmathtruncatemacro{\x}{\j +1}; \draw[-,medge] (a\j) -- (a\x);}
      \foreach \j in {2,4,6}{\pgfmathtruncatemacro{\x}{\j +1}; \draw[-,oedge] (a\j) -- (a\x);}

      \end{scope}

      \begin{scope}[xshift=5.3*\xr cm, yshift=-0.75*\yr cm]
        \draw[rounded corners, lightgray] (0,0) rectangle (5.4*\xr,1.5*\yr);
        \node (l1) at (1*\xr,1.25*\yr)[lnode]{required edges~$R$};
        \node (l2) at (1*\xr,0.75*\yr)[lnode]{added matching edges~$M^*$};
        \node (l3) at (1*\xr,0.25*\yr)[lnode]{optimal \EE{}~$D$};
        \draw[redge] ($(l1)-(2.5*\xr,0)$) -- ($(l1)-(3.1*\xr,0)$);
        \draw[medge] ($(l2)-(2.5*\xr,0)$) -- ($(l2)-(3.1*\xr,0)$);
        \draw[oedge] ($(l3)-(2.5*\xr,0)$) -- ($(l3)-(3.1*\xr,0)$);
      \end{scope}
\end{tikzpicture}
    \caption{Example showing that
      the bound given in \cref{obs:mstar}\eqref{match2} is tight:
      adding the edges in~$M^*$ to~$R$
      breaks the only optimal \EE{}~$D$ (dashed).
      To fix it,
      one either has to double all edges of~$D$
      or add all edges of~$M^*$ to~$D$.
      Note that the star can be arbitrarily enlarged.}
    \label{fig:matchtight}
  \end{figure}
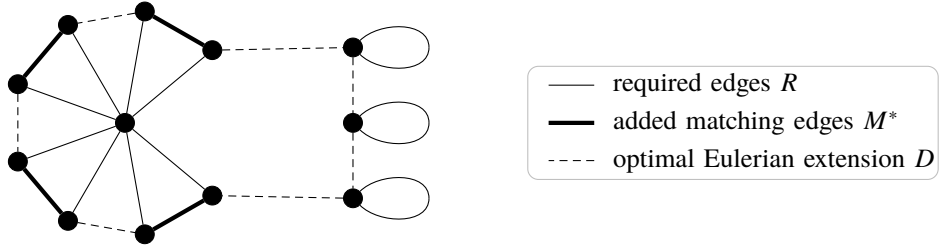
\item \cref{rrule:matching} increases the total weight
  of required edges.
  This makes it unusable for a PSAKS,
  since, in the resulting instance,
  a solution might be $(1+\varepsilon)$\hyp approximate
  merely due to the fact that the
  lower bound~\(\cost(R)\) on the solution is sufficiently large
  (we will use this fact below).
\end{compactenum}
  
\noindent
Given the difficulties of showing
a PSAKS for~$c$,
it is tempting to disprove its existence.
However,
we can easily build a 
PSAKS with size polynomial in~$\cost(T)$,
which gives a PSAKS of size polynomial in~$c$
in case that the edge weights are bounded by~$\poly(c)$.
More specifically,
we prove the following.

\begin{proposition}
  \label{thm:easy}
  Let $(G,R,\cost)$~be an instance of RPP
  with triangle inequality.
\begin{compactenum}[(i)]
\item\label{lightT} If $\cost(T) \le \varepsilon( \cost(R)+\cost(M))$,
  then
  a $(1+ 2\varepsilon)$-approximate \sol{}
  for $(G,R,\cost)$ can be found in polynomial time.
  
\item\label{lightM} If $\cost(M) \le \varepsilon( \cost(R)+\cost(T))$,
  then \((G,R,\cost)\)~has
  a $(1+3\varepsilon+2\varepsilon^2)$-approximate kernel
  with $O(c/\varepsilon)$~vertices.
  
\item\label{heavyboth}
  Otherwise,
  $(G,R,\cost)$~has
  an (exact) problem kernel with respect
  to the parameter $\min\{\cost(T)/\varepsilon-\cost(M),\cost(M)/\varepsilon -\cost(T)\}$.
\end{compactenum}
\end{proposition}

\noindent
\cref{thm:easy} shows that,
in order to exclude PSAKSes for RPP parameterized by~$c$,
a reduction must use unbounded edge weights,
the weights of~$T$, $M$, and~$R$
may not differ too much (by \eqref{lightT} and \eqref{lightM}),
yet the weights of~$T$ and~$M$
must not be too close either (by \eqref{heavyboth}).
Given these restrictions,
we conjecture:
\begin{conjecture}
  \label{mainconj}
  RPP has a PSAKS with respect to the parameter~$c$.
\end{conjecture}

\noindent
We finally prove \cref{thm:easy}.

\begin{proof}[Proof of \cref{thm:easy}]
  (\ref{lightT})
  Observe that the multiset $T\uplus T\uplus M$~is
  an \EE{} for~$(G,R,\cost)$.
  Using \cref{prop:linearlift},
  it yields an \sol{} of weight
  \begin{align*}
    \cost(R)+\cost(M)+2\omega(T)&\leq \cost(R)+\cost(M)+2\varepsilon(\omega(R)+\omega(M))\\
                                &\leq (1+2\varepsilon)(\cost(R)+\cost(D))&\text{using \eqref{cMD}}\hphantom.\\
    &=(1+2\varepsilon)\cost(W^*)&\text{using \eqref{cW=R+D}}.
  \end{align*}

  (\ref{lightM}) Let $R'$ be obtained from $R$
  using \cref{rrule:matching}
  with \(\pM=\cost(M)\),
  that is, $R'=R\uplus M$.
  In $\eG{G}{R'}$, all vertices are balanced.
  Thus, applying \cref{thm:psaks}
  to $(G,R',\cost)$
  gives an instance $(G_2,R_2,\cost_2)$
with 
  $O(c/\varepsilon)$~vertices.

  \looseness=-1
  Let $D$~be an optimal \EE{} for $(G,R,\cost)$.
  Then,
  by \cref{obs:mstar},
  an optimal \EE{}~$D'$ for $(G,R',\cost)$
  has weight $\cost(D')\leq \cost(D)+\delta=\cost(D)+\cost(M)$
  and,
  by \cref{prop:linearlift},
  an optimal \sol{} for~$(G,R',\cost)$ has weight $\cost(R')+\cost(D')$.
  Moreover,
  by \cref{thm:psaks},
any $\alpha$-approximate
  \sol{} for $(G_2,R_2,\cost_2)$
  can be lifted to an $\alpha(1+\varepsilon)$-approximate \sol{}~$W$ for~$(G,R',\cost)$.
By \cref{obs:mstar}\eqref{match1},
  it yields a \sol{} for~$(G,R,\cost)$ of weight
\begin{align*}
    \cost(W)\leq \alpha(1+\varepsilon)(\cost(R')+\cost(D'))& \le
                                                             \alpha(1+\varepsilon)(\cost(R) + 2\cost(M) +\cost(D))\\
                                                           &\le \alpha(1+\varepsilon)(\cost(R) + 2\varepsilon(\cost(R)+\cost(T))+\cost(D))&\text{using \eqref{lightM}}\phantom{.}\\
        &\le \alpha(1+\varepsilon)((1+2\varepsilon)\cost(R) + (1+2\varepsilon)\cost(D))&\text{using \eqref{cTD}}\\
    &=\alpha(1+\varepsilon)(1+2\varepsilon)(\cost(R)+\cost(D))\\
    &=\alpha(1+3\varepsilon+2\varepsilon^2)\cost(W^*)&\text{using \eqref{cW=R+D}.}
  \end{align*}

  \eqref{heavyboth}
  Otherwise, one has
  \begin{align*}
    \cost(R)&\leq \cost(M)/\varepsilon-\cost(T)&&\text{ and }&
    \cost(R)&\leq \cost(T)/\varepsilon-\cost(M)
  \end{align*}
  and thus the known
  $2|R|$-vertex problem kernel \citep{BNSW15}
  for RPP
  will be a kernel for both of these parameters.
\end{proof}

\section{Experiments}
\label{sec:exp}

In this section,
we experimentally evaluate the polynomial-size approximate kernelization scheme
presented in \cref{sec:psaksproof}.

\paragraph{Data instances.}
We evaluate the data reduction effect of our algorithm
on three data sets
generated by \citet{CLS01,CPS07}:
\begin{compactitem}[madr-$p$-$i$:]
\item[alba-$p$-$i$] for each $p\in\{0.3,0.5,0.7\}$ and $i\in\{1,\dots,5\}$:
  based on the street network of the Spanish town Albaida,
  where each edge is required with probability~$p$
  and $i$~is just a running index.
\item[madr-$p$-$i$] for each $p\in\{0.3,0.5,0.7\}$ and $i\in\{1,\dots,5\}$:
  based on the street network of the Spanish town Madrigueras,
  where each edge is required with probability~$p$
  and $i$~is just a running index.
\item[ur-$n$-$d$-$p$] for each $n\in\{500,750,1000\}$,  $d\in\{3,4,5,6\}$, and $p\in\{0.25,0.5,0.75\}$:
  $n$~vertices are selected randomly from an $(1000\times 1000)$-grid,
  distances are Euclidean,
  each vertex is connected to its $d$~closest neighbors,
  and
  each edge is required with probability~$p$.
\end{compactitem}
These data sets  are widely used in the literature~\citep{CLS01,CPS07,FMGO03,HBLR14,RTW08}.\footnote{Available at \url{https://www.uv.es/corberan/instancias.htm}}
We also test our algorithm on two instances provided to us
by Berliner Stadtreinigung,
the company responsible for snow plowing and garbage collection in Berlin.\footnote{Available at \url{https://gitlab.com/rvb/rpp-psaks}}
In the Berlin instances,
both the street network as well as the required edges
arise from a real snow plowing application,
as opposed to generating the required edges randomly
like in the other instances.

Characteristics of these instances
and the weight $\cost(W)$ of a solution obtained via the 3/2-approximation
can be seen in the ``input instance'' columns of
\cref{tab:alba,tab:real,tab:berlin,tab:random}.

\begin{table}[p]
  \small
  \centering
  \caption{Results on the alba-$p$-$i$ instances.
    Highlighted are rows
    where the weight of an approximate solution for the input instance
    differs from the weight of an approximate solution lifted from the
    kernel.}

  \begin{tabular}{>{\rowmac}r|>{\rowmac}r>{\rowmac}r>{\rowmac}r>{\rowmac}r>{\rowmac}r>{\rowmac}r|>{\rowmac}r>{\rowmac}r>{\rowmac}r>{\rowmac}r|>{\rowmac}r>{\rowmac}r>{\rowmac}r>{\rowmac}l<{\clearrow}}
    \toprule
    &\multicolumn{6}{c|}{input instance}
    &\multicolumn{4}{c|}{kernel}
    &\multicolumn{3}{c}{comparison}\\
$p$ & $|V|$ & $|V(R)|$ & $|R|$ & $b$ & $c$ & $\cost(W)$ & $|V'|$ & $|R'|$ & $\cost(W')$ & ms & $\frac{|V'|}{|V|}$ & $\frac{|V'|}{|V(R)|}$ & $\frac{|R'|}{|R|}$ & $\frac{\cost(W')}{\cost(W)}$ \\ \midrule 
0.3 & 116   & 72       & 51    & 54  & 22  & 7987       & 72     & 51     & 7987        & 2  & 0.62               & 1.00                  & 1.00               & 1                         \\
0.3 & 116   & 68       & 46    & 54  & 23  & 6950       & 68     & 46     & 6950        & 2  & 0.59               & 1.00                  & 1.00               & 1                         \\
0.3 & 116   & 59       & 44    & 36  & 15  & 7587       & 59     & 44     & 7587        & 2  & 0.51               & 1.00                  & 1.00               & 1                         \\
0.3 & 116   & 70       & 49    & 48  & 21  & 7464       & 70     & 49     & 7464        & 2  & 0.60               & 1.00                  & 1.00               & 1                         \\
0.3 & 116   & 73       & 57    & 48  & 19  & 7972       & 73     & 57     & 7972        & 2  & 0.63               & 1.00                  & 1.00               & 1                         \\
0.5 & 116   & 101      & 88    & 68  & 18  & 11387      & 101    & 88     & 11387       & 3  & 0.87               & 1.00                  & 1.00               & 1                         \\
0.5 & 116   & 100      & 92    & 58  & 14  & 10796      & 100    & 92     & 10796       & 4  & 0.86               & 1.00                  & 1.00               & 1                         \\
0.5 & 116   & 99       & 92    & 50  & 11  & 9469       & 98     & 91     & 9469        & 4  & 0.84               & 0.99                  & 0.99               & 1                         \\
0.5 & 116   & 91       & 88    & 50  & 8   & 9050       & 88     & 85     & 9050        & 5  & 0.76               & 0.97                  & 0.97               & 1                         \\
0.5 & 116   & 102      & 91    & 60  & 16  & 10137      & 102    & 91     & 10137       & 3  & 0.88               & 1.00                  & 1.00               & 1                         \\
\setrow{\bfseries}0.7 & 116   & 104      & 118   & 64  & 6   & 11521      & 89     & 95     & 11641       & 12 & 0.77               & 0.86                  & 0.81               & 1.01                         \\
0.7 & 116   & 108      & 122   & 56  & 2   & 11155      & 58     & 65     & 11155       & 30 & 0.50               & 0.54                  & 0.53               & 1                         \\
0.7 & 116   & 110      & 113   & 60  & 9   & 11895      & 104    & 107    & 11895       & 8  & 0.90               & 0.95                  & 0.95               & 1                         \\
0.7 & 116   & 110      & 119   & 66  & 4   & 11761      & 83     & 88     & 11761       & 25 & 0.72               & 0.75                  & 0.74               & 1                         \\
0.7 & 116   & 110      & 116   & 58  & 7   & 11414      & 96     & 102    & 11414       & 13 & 0.83               & 0.87                  & 0.88               & 1                         \\
    \bottomrule
  \end{tabular}
  \label{tab:alba}
\end{table}
\begin{table}[p]
  \small
  \centering
  \caption{Results on the madr-$p$-$i$ instances.
      Highlighted are rows
    where the weight of an approximate solution for the input instances
    differs from the weight of an approximate solution lifted from the
    kernel.}
\begin{tabular}{>{\rowmac}r|>{\rowmac}r>{\rowmac}r>{\rowmac}r>{\rowmac}r>{\rowmac}r>{\rowmac}r|>{\rowmac}r>{\rowmac}r>{\rowmac}r>{\rowmac}r|>{\rowmac}r>{\rowmac}r>{\rowmac}r>{\rowmac}l<{\clearrow}}
    \toprule
    &\multicolumn{6}{c|}{input instance}
    &\multicolumn{4}{c|}{kernel}
    &\multicolumn{4}{c}{comparison}\\
    $p$ & $|V|$ & $|V(R)|$ & $|R|$ & $b$ & $c$ & $\cost(W)$ & $|V'|$ & $|R'|$ & $\cost(W')$ & ms & $\frac{|V'|}{|V|}$ & $\frac{|V'|}{|V(R)|}$ & $\frac{|R'|}{|R|}$ & $\frac{\cost(W')}{\cost(W)}$ \\ \midrule 
 0.3    & 196   & 127     & 86    & 96  & 42  & 13\,090      & 127    & 86     & 13\,090       & 4  & 0.65                   & 1.00                 & 1.00               &1                         \\
 0.3    & 196   & 142     & 108   & 86  & 34  & 14\,220      & 142    & 108    & 14\,220       & 5  & 0.72                   & 1.00                 & 1.00               &1                         \\
 0.3    & 196   & 137     & 102   & 96  & 36  & 13\,510      & 137    & 102    & 13\,510       & 4  & 0.70                   & 1.00                 & 1.00               &1                         \\
 0.3    & 196   & 140     & 101   & 98  & 39  & 13\,765      & 140    & 101    & 13\,765       & 4  & 0.71                   & 1.00                 & 1.00               &1                         \\
 0.3    & 196   & 131     & 95    & 88  & 38  & 13\,275      & 131    & 95     & 13\,275       & 4  & 0.67                   & 1.00                 & 1.00               &1                         \\
 0.5    & 196   & 176     & 163   & 108 & 21  & 15\,780      & 176    & 163    & 15\,780       & 8  & 0.90                   & 1.00                 & 1.00               &1                         \\
 0.5    & 196   & 174     & 156   & 100 & 25  & 17\,120      & 174    & 156    & 17\,120       & 9  & 0.89                   & 1.00                 & 1.00               &1                         \\
 0.5    & 196   & 165     & 148   & 94  & 22  & 15\,465      & 165    & 148    & 15\,465       & 9  & 0.84                   & 1.00                 & 1.00               &1                         \\
 0.5    & 196   & 166     & 152   & 92  & 23  & 16\,920      & 166    & 152    & 16\,920       & 7  & 0.85                   & 1.00                 & 1.00               &1                         \\
 0.5    & 196   & 169     & 147   & 96  & 26  & 15\,835      & 169    & 147    & 15\,835       & 6  & 0.86                   & 1.00                 & 1.00               &1                         \\
\setrow{\bfseries} 0.7    & 196   & 188     & 211   & 96  & 7   & 20\,660      & 124    & 134    & 20\,560       & 67 & 0.63                   & 0.66                 & 0.64               &0.995                         \\
 0.7    & 196   & 192     & 238   & 120 & 2   & 22\,220      & 123    & 151    & 22\,220       & 81 & 0.63                   & 0.64                 & 0.63               &1                         \\
 0.7    & 196   & 191     & 219   & 92  & 6   & 20\,785      & 118    & 132    & 20\,785       & 86 & 0.60                   & 0.62                 & 0.60               &1                         \\
 0.7    & 196   & 192     & 225   & 98  & 3   & 20\,815      & 103    & 123    & 20\,815       & 88 & 0.53                   & 0.54                 & 0.55               &1                         \\
\setrow{\bfseries} 0.7    & 196   & 191     & 223   & 106 & 3   & 21\,150      & 110    & 124    & 21\,250       & 87 & 0.56                   & 0.58                 & 0.56               &1.005                         \\
    \bottomrule
\end{tabular}

\label{tab:real}
\end{table}

\begin{table}[p]
  \small
  \centering
  \caption{Results on the instances from Berliner Stadtreinigung.}

  \begin{tabular}{rrrrrr|rrrr|rrrl}
    \toprule
    \multicolumn{6}{c|}{input instance}
    &\multicolumn{4}{c|}{kernel}
    &\multicolumn{3}{c}{comparison}\\
    $|V|$ & $|V(R)|$ & $|R|$ & $b$ & $c$ & $\cost(W)$ & $|V'|$ & $|R'|$ & $\cost(W')$ & ms  & $\frac{|V'|}{|V|}$ & $\frac{|V'|}{|V(R)|}$ & $\frac{|R'|}{|R|}$ & $\frac{\cost(W')}{\cost(W)}$ \\ \midrule 
    2\,593  & 285     & 289   & 34  & 3   & 21\,911      & 62     & 66     & 21\,911       & 263 & 0.02& 0.22                 & 0.23               & 1                         \\
    5\,097  & 369     & 408   & 56  & 3   & 31\,694      & 70     & 82     & 31\,694       & 435 & 0.01& 0.19                 & 0.20               & 1                         \\
    \bottomrule
  \end{tabular}
  \label{tab:berlin}
\end{table}

\begin{table}[p]
  \small
  \centering
  \caption{Results on the ur-$n$-$d$-$p$ instances.  In these instances, $|V|=|V(R)|$.  Highlighted are rows
    where the weight of an approximate solution for the input instances
    differs from the weight of an approximate solution lifted from the
    kernel.}

  \begin{tabular}{>{\rowmac}r>{\rowmac}r>{\rowmac}r|>{\rowmac}r>{\rowmac}r>{\rowmac}r>{\rowmac}r>{\rowmac}r|>{\rowmac}r>{\rowmac}r>{\rowmac}r>{\rowmac}r|>{\rowmac}r>{\rowmac}r>{\rowmac}l>{\rowmac}l<{\clearrow}}
  \toprule
  \multicolumn{3}{c|}{parameters}
  &\multicolumn{5}{c|}{input instance}
  &\multicolumn{4}{c|}{kernel}
  &\multicolumn{4}{c}{comparison}\\
   $n$ & $d$ & $p$  & $|V(R)|$ & $|R|$ & $b$ & $c$ & $\cost(W)$ & $|V'|$ & $|R'|$ & $\cost(W')$ & ms     & $\frac{|V'|}{|V(R)|}$ & $\frac{|R'|}{|R|}$ & $\frac{\cost(W')}{\cost(W)}$ & $\frac{\cost(W')}{\text{opt}}$ \\\midrule
500    & 3   & 0.25 & 298      & 206   & 218 & 99  & 18\,004    & 298    & 206    & 18\,004     & 9      & 1.00                  & 1.00               & 1                            & 1.0421                         \\
500    & 3   & 0.50 & 458      & 464   & 246 & 58  & 24\,249    & 449    & 454    & 24\,249     & 38     & 0.98                  & 0.98               & 1                            & 1.0260                         \\
\setrow{\bfseries}500    & 3   & 0.75 & 493      & 671   & 246 & 19  & 30\,141    & 338    & 438    & 30\,161     & 336    & 0.69                  & 0.65               & 1.0007                       & 1.0021                         \\
500    & 4   & 0.25 & 343      & 268   & 216 & 85  & 19\,152    & 343    & 268    & 19\,152     & 12     & 1.00                  & 1.00               & 1                            & 1.0741                         \\
\setrow{\bfseries}500    & 4   & 0.50 & 476      & 582   & 242 & 19  & 29\,865    & 346    & 400    & 29\,845     & 337    & 0.73                  & 0.69               & 0.9993                       & 1.0066                         \\
500    & 4   & 0.75 & 498      & 848   & 242 & 2   & 38\,692    & 244    & 339    & 38\,692     & 644    & 0.49                  & 0.40               & 1                            & 1                              \\
500    & 5   & 0.25 & 388      & 322   & 238 & 80  & 21\,124    & 387    & 321    & 21\,124     & 30     & 1.00                  & 1.00               & 1                            & 1.0511                         \\
\setrow{\bfseries}500    & 5   & 0.50 & 490      & 672   & 242 & 5   & 34\,560    & 265    & 334    & 34\,524     & 650    & 0.54                  & 0.50               & 0.9990                       & 1.0010                         \\
500    & 5   & 0.75 & 498      & 1001  & 252 & 1   & 48\,307    & 255    & 377    & 48\,307     & 543    & 0.51                  & 0.38               & 1                            & 1                              \\
500    & 6   & 0.25 & 416      & 405   & 232 & 53  & 25\,214    & 406    & 392    & 25\,214     & 39     & 0.98                  & 0.97               & 1                            & 1.0268                         \\
\setrow{\bfseries}500    & 6   & 0.50 & 496      & 793   & 248 & 2   & 42\,845    & 256    & 357    & 42\,853     & 648    & 0.52                  & 0.45               & 1.0002                       & 1.0006                         \\
500    & 6   & 0.75 & 499      & 1157  & 250 & 1   & 58\,971    & 250    & 396    & 58\,971     & 570    & 0.50                  & 0.34               & 1                            & 1                              \\
700    & 3   & 0.25 & 452      & 321   & 328 & 140 & 22\,114    & 451    & 320    & 22\,114     & 15     & 1.00                  & 1.00               & 1                            & 1.0474                         \\
700    & 3   & 0.50 & 662      & 648   & 378 & 100 & 29\,289    & 651    & 635    & 29\,288     & 63     & 0.98                  & 0.98               & 1                            & 1.0218                         \\
\setrow{\bfseries}700    & 3   & 0.75 & 744      & 979   & 390 & 16  & 36\,588    & 423    & 540    & 36\,732     & 971    & 0.57                  & 0.55               & 1.0039                       & 1.0039                         \\
700    & 4   & 0.25 & 538      & 439   & 340 & 122 & 24\,084    & 536    & 437    & 24\,084     & 22     & 1.00                  & 1.00               & 1                            & 1.0677                         \\
\setrow{\bfseries}700    & 4   & 0.50 & 713      & 808   & 378 & 57  & 32\,830    & 655    & 733    & 32\,857     & 229    & 0.92                  & 0.91               & 1.0008                       & 1.0112                         \\
\setrow{\bfseries}700    & 4   & 0.75 & 745      & 1261  & 356 & 3   & 47\,769    & 366    & 498    & 47\,774     & 1\,486 & 0.49                  & 0.39               & 1.0001                       & 1.0002                         \\
\setrow{\bfseries}700    & 5   & 0.25 & 580      & 506   & 344 & 108 & 26\,315    & 577    & 503    & 26\,317     & 27     & 0.99                  & 0.99               & 1.0001                       & 1.0472                         \\
\setrow{\bfseries}700    & 5   & 0.50 & 724      & 1003  & 398 & 15  & 41\,897    & 418    & 521    & 41\,946     & 1\,012 & 0.58                  & 0.52               & 1.0012                       & 1.0041                         \\
700    & 5   & 0.75 & 748      & 1459  & 380 & 1   & 58\,416    & 388    & 592    & 58\,416     & 1\,145 & 0.52                  & 0.41               & 1                            & 1                              \\
700    & 6   & 0.25 & 593      & 530   & 360 & 103 & 28\,920    & 591    & 528    & 28\,920     & 27     & 1.00                  & 1.00               & 1                            & 1.0373                         \\
\setrow{\bfseries}700    & 6   & 0.50 & 741      & 1179  & 376 & 2   & 50\,492    & 385    & 528    & 50\,508     & 1\,391 & 0.52                  & 0.45               & 1.0003                       & 1.0003                         \\
700    & 6   & 0.75 & 749      & 1747  & 396 & 1   & 72\,950    & 397    & 615    & 72\,950     & 1\,255 & 0.53                  & 0.35               & 1                            & 1                              \\
1\,000 & 3   & 0.25 & 605      & 411   & 442 & 204 & 25\,460    & 605    & 411    & 25\,460     & 19     & 1.00                  & 1.00               & 1                            & 1.0647                         \\
1\,000 & 3   & 0.50 & 892      & 892   & 502 & 124 & 33\,981    & 865    & 862    & 33\,981     & 146    & 0.97                  & 0.97               & 1                            & 1.0270                         \\
\setrow{\bfseries}1\,000 & 3   & 0.75 & 980      & 1308  & 514 & 24  & 42\,894    & 585    & 739    & 43\,176     & 1\,809 & 0.60                  & 0.56               & 1.0066                       & 1.0089                         \\
1\,000 & 4   & 0.25 & 709      & 564   & 466 & 167 & 27\,290    & 703    & 558    & 27\,290     & 40     & 0.99                  & 0.99               & 1                            & 1.0682                         \\
\setrow{\bfseries}1\,000 & 4   & 0.50 & 929      & 1086  & 506 & 71  & 39\,607    & 856    & 983    & 39\,609     & 344    & 0.92                  & 0.91               & 1.0001                       & 1.0154                         \\
\setrow{\bfseries}1\,000 & 4   & 0.75 & 996      & 1684  & 496 & 4   & 55\,967    & 514    & 694    & 56\,010     & 2\,689 & 0.52                  & 0.41               & 1.0008                       & 1.0009                         \\
\setrow{\bfseries}1\,000 & 5   & 0.25 & 766      & 661   & 488 & 149 & 30\,464    & 765    & 660    & 30\,467     & 30     & 1.00                  & 1.00               & 1.0001                       & 1.0515                         \\
\setrow{\bfseries}1\,000 & 5   & 0.50 & 975      & 1352  & 494 & 5   & 49\,197    & 524    & 667    & 49\,216     & 2\,544 & 0.54                  & 0.49               & 1.0004                       & 1.0012                         \\
1\,000 & 5   & 0.75 & 1000     & 2029  & 516 & 2   & 70\,231    & 521    & 763    & 70\,231     & 2\,650 & 0.52                  & 0.38               & 1                            & 1                              \\
\setrow{\bfseries}1\,000 & 6   & 0.25 & 802      & 728   & 486 & 138 & 33\,688    & 792    & 717    & 33\,690     & 60     & 0.99                  & 0.98               & 1.0001                       & 1.0417                         \\
\setrow{\bfseries}1\,000 & 6   & 0.50 & 980      & 1563  & 510 & 9   & 58\,854    & 523    & 702    & 58\,951     & 2\,262 & 0.53                  & 0.45               & 1.0016                       & 1.0026                         \\
1\,000 & 6   & 0.75 & 1000     & 2304  & 502 & 1   & 82\,481    & 504    & 781    & 82\,481     & 2\,265 & 0.50                  & 0.34               & 1                            & 1                              \\
  \bottomrule
\end{tabular}
\label{tab:random}
\end{table}

\paragraph{Experimental setup.}
Since our main goal is evaluating the effect of our data reduction
rather than the running time of our algorithm,
we sacrificed speed for simplicity and implemented
the part of our PSAKS described in \cref{step:shrinkgraph}
in approximately 200 lines of Python
(not counting the testing environment)
using the NetworkX library for
finding minimum\hyp weight perfect matchings,
(bi)connected components,
cut vertices, and
spanning trees.\footnote{\url{https://networkx.github.io/}}
These routines are also contained in
highly optimized C++ libraries like LEMON\footnote{\url{https://lemon.cs.elte.hu/trac/lemon}}
and we expect that one could achieve a speedup by orders of magnitude
by implementing our PSAKS in C++.
We did not implement the
weight reduction step described in \cref{step:weights},
since it is mainly of theoretical interest
(to prove a polynomial \emph{size} of the kernel
rather than just a polynomial number of vertices and edges).

We kernelized each of the instances listed above
for $\varepsilon=1/10$, that is,
we require that a $11\alpha/10$-approximation
be recoverable from an $\alpha$-approximate solution in the kernel.
Since we do not reduce weights,
this means we apply \cref{rrule:ballred} with
$\varepsilon_1=\varepsilon=1/10$ in \eqref{gamma}.

We also apply
the folklore 3/2-approximation algorithm
based on the \citeauthor{Chr76}-\citeauthor{Ser78}
algorithm for the metric Traveling Salesman Problem \citep{Chr76,Ser78,BSxxb}
to compute a solution
in the original and kernelized instance
and compare their weights.

\begin{figure*}
  \centering
  \begin{tikzpicture}[x=1pt,y=1pt]
\definecolor{fillColor}{RGB}{255,255,255}
\path[use as bounding box,fill=fillColor,fill opacity=0.00] (0,0) rectangle (216.81,144.54);
\begin{scope}
\path[clip] ( 33.60, 25.20) rectangle (216.81,142.86);
\definecolor{drawColor}{RGB}{0,0,0}

\path[draw=drawColor,line width= 1.2pt,line join=round] ( 60.74, 87.23) -- ( 74.31, 87.23);

\path[draw=drawColor,line width= 0.4pt,dash pattern=on 4pt off 4pt ,line join=round,line cap=round] ( 67.53, 31.69) -- ( 67.53, 55.19);

\path[draw=drawColor,line width= 0.4pt,dash pattern=on 4pt off 4pt ,line join=round,line cap=round] ( 67.53,117.14) -- ( 67.53,106.46);

\path[draw=drawColor,line width= 0.4pt,line join=round,line cap=round] ( 64.14, 31.69) -- ( 70.92, 31.69);

\path[draw=drawColor,line width= 0.4pt,line join=round,line cap=round] ( 64.14,117.14) -- ( 70.92,117.14);

\path[draw=drawColor,line width= 0.4pt,line join=round,line cap=round] ( 60.74, 55.19) --
	( 74.31, 55.19) --
	( 74.31,106.46) --
	( 60.74,106.46) --
	( 60.74, 55.19);

\path[draw=drawColor,line width= 1.2pt,line join=round] (128.60, 74.42) -- (142.17, 74.42);

\path[draw=drawColor,line width= 0.4pt,dash pattern=on 4pt off 4pt ,line join=round,line cap=round] (135.38, 38.10) -- (135.38, 59.46);

\path[draw=drawColor,line width= 0.4pt,dash pattern=on 4pt off 4pt ,line join=round,line cap=round] (135.38,117.14) -- (135.38,105.39);

\path[draw=drawColor,line width= 0.4pt,line join=round,line cap=round] (131.99, 38.10) -- (138.78, 38.10);

\path[draw=drawColor,line width= 0.4pt,line join=round,line cap=round] (131.99,117.14) -- (138.78,117.14);

\path[draw=drawColor,line width= 0.4pt,line join=round,line cap=round] (128.60, 59.46) --
	(142.17, 59.46) --
	(142.17,105.39) --
	(128.60,105.39) --
	(128.60, 59.46);

\path[draw=drawColor,line width= 1.2pt,line join=round] (196.45, 76.55) -- (210.02, 76.55);

\path[draw=drawColor,line width= 0.4pt,dash pattern=on 4pt off 4pt ,line join=round,line cap=round] (203.24, 29.56) -- (203.24, 35.97);

\path[draw=drawColor,line width= 0.4pt,dash pattern=on 4pt off 4pt ,line join=round,line cap=round] (203.24,138.50) -- (203.24,136.37);

\path[draw=drawColor,line width= 0.4pt,line join=round,line cap=round] (199.85, 29.56) -- (206.63, 29.56);

\path[draw=drawColor,line width= 0.4pt,line join=round,line cap=round] (199.85,138.50) -- (206.63,138.50);

\path[draw=drawColor,line width= 0.4pt,line join=round,line cap=round] (196.45, 35.97) --
	(210.02, 35.97) --
	(210.02,136.37) --
	(196.45,136.37) --
	(196.45, 35.97);
\end{scope}
\begin{scope}
\path[clip] (  0.00,  0.00) rectangle (216.81,144.54);
\definecolor{drawColor}{RGB}{0,0,0}

\path[draw=drawColor,line width= 0.4pt,line join=round,line cap=round] ( 33.60, 31.69) -- ( 33.60,138.50);

\path[draw=drawColor,line width= 0.4pt,line join=round,line cap=round] ( 33.60, 31.69) -- ( 29.40, 31.69);

\path[draw=drawColor,line width= 0.4pt,line join=round,line cap=round] ( 33.60, 53.06) -- ( 29.40, 53.06);

\path[draw=drawColor,line width= 0.4pt,line join=round,line cap=round] ( 33.60, 74.42) -- ( 29.40, 74.42);

\path[draw=drawColor,line width= 0.4pt,line join=round,line cap=round] ( 33.60, 95.78) -- ( 29.40, 95.78);

\path[draw=drawColor,line width= 0.4pt,line join=round,line cap=round] ( 33.60,117.14) -- ( 29.40,117.14);

\path[draw=drawColor,line width= 0.4pt,line join=round,line cap=round] ( 33.60,138.50) -- ( 29.40,138.50);

\node[text=drawColor,rotate= 90.00,anchor=base,inner sep=0pt, outer sep=0pt, scale=  1.00] at ( 23.52, 31.69) {0.5};

\node[text=drawColor,rotate= 90.00,anchor=base,inner sep=0pt, outer sep=0pt, scale=  1.00] at ( 23.52, 53.06) {0.6};

\node[text=drawColor,rotate= 90.00,anchor=base,inner sep=0pt, outer sep=0pt, scale=  1.00] at ( 23.52, 74.42) {0.7};

\node[text=drawColor,rotate= 90.00,anchor=base,inner sep=0pt, outer sep=0pt, scale=  1.00] at ( 23.52, 95.78) {0.8};

\node[text=drawColor,rotate= 90.00,anchor=base,inner sep=0pt, outer sep=0pt, scale=  1.00] at ( 23.52,117.14) {0.9};

\node[text=drawColor,rotate= 90.00,anchor=base,inner sep=0pt, outer sep=0pt, scale=  1.00] at ( 23.52,138.50) {1.0};
\end{scope}
\begin{scope}
\path[clip] (  0.00,  0.00) rectangle (216.81,144.54);
\definecolor{drawColor}{RGB}{0,0,0}

\node[text=drawColor,rotate= 90.00,anchor=base,inner sep=0pt, outer sep=0pt, scale=  1.00] at ( 10.92, 84.03) {$|V'|/|V|$};
\end{scope}
\begin{scope}
\path[clip] (  0.00,  0.00) rectangle (216.81,144.54);
\definecolor{drawColor}{RGB}{0,0,0}

\path[draw=drawColor,line width= 0.4pt,line join=round,line cap=round] ( 33.60, 25.20) --
	(216.81, 25.20) --
	(216.81,142.86) --
	( 33.60,142.86) --
	( 33.60, 25.20);
\end{scope}
\begin{scope}
\path[clip] ( 33.60, 25.20) rectangle (216.81,142.86);
\definecolor{drawColor}{RGB}{190,190,190}

\path[draw=drawColor,line width= 0.4pt,dash pattern=on 1pt off 3pt ,line join=round,line cap=round] ( 91.28, 25.20) -- ( 91.28,142.86);

\path[draw=drawColor,line width= 0.4pt,dash pattern=on 1pt off 3pt ,line join=round,line cap=round] (159.13, 25.20) -- (159.13,142.86);

\path[draw=drawColor,line width= 0.4pt,dash pattern=on 1pt off 3pt ,line join=round,line cap=round] ( 33.60, 31.69) -- (216.81, 31.69);

\path[draw=drawColor,line width= 0.4pt,dash pattern=on 1pt off 3pt ,line join=round,line cap=round] ( 33.60, 53.06) -- (216.81, 53.06);

\path[draw=drawColor,line width= 0.4pt,dash pattern=on 1pt off 3pt ,line join=round,line cap=round] ( 33.60, 74.42) -- (216.81, 74.42);

\path[draw=drawColor,line width= 0.4pt,dash pattern=on 1pt off 3pt ,line join=round,line cap=round] ( 33.60, 95.78) -- (216.81, 95.78);

\path[draw=drawColor,line width= 0.4pt,dash pattern=on 1pt off 3pt ,line join=round,line cap=round] ( 33.60,117.14) -- (216.81,117.14);
\definecolor{fillColor}{RGB}{0,0,0}

\path[fill=fillColor] ( 42.20, 56.83) rectangle ( 43.20, 57.83);

\path[fill=fillColor] ( 47.01, 50.42) rectangle ( 48.01, 51.42);

\path[fill=fillColor] ( 46.35, 33.33) rectangle ( 47.35, 34.33);

\path[fill=fillColor] ( 40.63, 52.56) rectangle ( 41.63, 53.56);

\path[fill=fillColor] ( 44.22, 58.96) rectangle ( 45.22, 59.96);

\path[fill=fillColor] ( 42.76,110.23) rectangle ( 43.76,111.23);

\path[fill=fillColor] ( 46.17,108.10) rectangle ( 47.17,109.10);

\path[fill=fillColor] ( 48.01,103.82) rectangle ( 49.01,104.82);

\path[fill=fillColor] ( 42.56, 86.73) rectangle ( 43.56, 87.73);

\path[fill=fillColor] ( 50.58,112.37) rectangle ( 51.58,113.37);

\path[fill=fillColor] ( 52.51, 88.87) rectangle ( 53.51, 89.87);

\path[fill=fillColor] ( 48.23, 31.19) rectangle ( 49.23, 32.19);

\path[fill=fillColor] ( 45.04,116.64) rectangle ( 46.04,117.64);

\path[fill=fillColor] ( 52.18, 78.19) rectangle ( 53.18, 79.19);

\path[fill=fillColor] ( 42.12,101.69) rectangle ( 43.12,102.69);

\path[fill=fillColor] (120.35, 63.24) rectangle (121.35, 64.24);

\path[fill=fillColor] (108.86, 78.19) rectangle (109.86, 79.19);

\path[fill=fillColor] (115.85, 73.92) rectangle (116.85, 74.92);

\path[fill=fillColor] (119.44, 76.05) rectangle (120.44, 77.05);

\path[fill=fillColor] (113.15, 67.51) rectangle (114.15, 68.51);

\path[fill=fillColor] (119.29,116.64) rectangle (120.29,117.64);

\path[fill=fillColor] (115.26,114.50) rectangle (116.26,115.50);

\path[fill=fillColor] (112.85,103.82) rectangle (113.85,104.82);

\path[fill=fillColor] (120.36,105.96) rectangle (121.36,106.96);

\path[fill=fillColor] (113.89,108.10) rectangle (114.89,109.10);

\path[fill=fillColor] (112.93, 58.96) rectangle (113.93, 59.96);

\path[fill=fillColor] (118.00, 58.96) rectangle (119.00, 59.96);

\path[fill=fillColor] (108.68, 52.56) rectangle (109.68, 53.56);

\path[fill=fillColor] (118.71, 37.60) rectangle (119.71, 38.60);

\path[fill=fillColor] (117.06, 44.01) rectangle (118.06, 45.01);

\path[fill=fillColor] (183.26,138.00) rectangle (184.26,139.00);

\path[fill=fillColor] (185.83,133.73) rectangle (186.83,134.73);

\path[fill=fillColor] (183.80, 71.78) rectangle (184.80, 72.78);

\path[fill=fillColor] (187.61,138.00) rectangle (188.61,139.00);

\path[fill=fillColor] (185.18, 80.33) rectangle (186.18, 81.33);

\path[fill=fillColor] (176.78, 29.06) rectangle (177.78, 30.06);

\path[fill=fillColor] (181.36,138.00) rectangle (182.36,139.00);

\path[fill=fillColor] (182.63, 39.74) rectangle (183.63, 40.74);

\path[fill=fillColor] (185.74, 33.33) rectangle (186.74, 34.33);

\path[fill=fillColor] (179.52,133.73) rectangle (180.52,134.73);

\path[fill=fillColor] (177.48, 35.47) rectangle (178.48, 36.47);

\path[fill=fillColor] (176.19, 31.19) rectangle (177.19, 32.19);

\path[fill=fillColor] (187.94,138.00) rectangle (188.94,139.00);

\path[fill=fillColor] (188.68,133.73) rectangle (189.68,134.73);

\path[fill=fillColor] (188.28, 46.15) rectangle (189.28, 47.15);

\path[fill=fillColor] (186.45,138.00) rectangle (187.45,139.00);

\path[fill=fillColor] (182.43,120.91) rectangle (183.43,121.91);

\path[fill=fillColor] (180.59, 29.06) rectangle (181.59, 30.06);

\path[fill=fillColor] (176.63,135.87) rectangle (177.63,136.87);

\path[fill=fillColor] (185.32, 48.28) rectangle (186.32, 49.28);

\path[fill=fillColor] (176.59, 35.47) rectangle (177.59, 36.47);

\path[fill=fillColor] (187.16,138.00) rectangle (188.16,139.00);

\path[fill=fillColor] (188.47, 35.47) rectangle (189.47, 36.47);

\path[fill=fillColor] (176.52, 37.60) rectangle (177.52, 38.60);

\path[fill=fillColor] (186.85,138.00) rectangle (187.85,139.00);

\path[fill=fillColor] (184.30,131.59) rectangle (185.30,132.59);

\path[fill=fillColor] (177.47, 52.56) rectangle (178.47, 53.56);

\path[fill=fillColor] (187.01,135.87) rectangle (188.01,136.87);

\path[fill=fillColor] (182.30,120.91) rectangle (183.30,121.91);

\path[fill=fillColor] (189.12, 35.47) rectangle (190.12, 36.47);

\path[fill=fillColor] (180.12,138.00) rectangle (181.12,139.00);

\path[fill=fillColor] (187.64, 39.74) rectangle (188.64, 40.74);

\path[fill=fillColor] (176.92, 35.47) rectangle (177.92, 36.47);

\path[fill=fillColor] (189.14,135.87) rectangle (190.14,136.87);

\path[fill=fillColor] (175.75, 37.60) rectangle (176.75, 38.60);

\path[fill=fillColor] (182.91, 31.19) rectangle (183.91, 32.19);
\end{scope}
\begin{scope}
\path[clip] (  0.00,  0.00) rectangle (216.81,144.54);
\definecolor{drawColor}{RGB}{0,0,0}

\path[draw=drawColor,line width= 0.4pt,line join=round,line cap=round] ( 57.35, 25.20) -- (193.06, 25.20);

\path[draw=drawColor,line width= 0.4pt,line join=round,line cap=round] ( 57.35, 25.20) -- ( 57.35, 21.00);

\path[draw=drawColor,line width= 0.4pt,line join=round,line cap=round] (125.20, 25.20) -- (125.20, 21.00);

\path[draw=drawColor,line width= 0.4pt,line join=round,line cap=round] (193.06, 25.20) -- (193.06, 21.00);

\node[text=drawColor,anchor=base,inner sep=0pt, outer sep=0pt, scale=  1.00] at ( 57.35, 10.08) {alba-$p$-$i$};

\node[text=drawColor,anchor=base,inner sep=0pt, outer sep=0pt, scale=  1.00] at (125.20, 10.08) {madr-$p$-$i$};

\node[text=drawColor,anchor=base,inner sep=0pt, outer sep=0pt, scale=  1.00] at (193.06, 10.08) {ur-$n$-$d$-$p$};
\end{scope}
\end{tikzpicture}
   \hfill
  \begin{tikzpicture}[x=1pt,y=1pt]
\definecolor{fillColor}{RGB}{255,255,255}
\path[use as bounding box,fill=fillColor,fill opacity=0.00] (0,0) rectangle (216.81,144.54);
\begin{scope}
\path[clip] ( 33.60, 25.20) rectangle (216.81,142.86);
\definecolor{drawColor}{RGB}{0,0,0}

\path[draw=drawColor,line width= 1.2pt,line join=round] ( 60.74,138.50) -- ( 74.31,138.50);

\path[draw=drawColor,line width= 0.4pt,dash pattern=on 4pt off 4pt ,line join=round,line cap=round] ( 67.53,108.60) -- ( 67.53,119.28);

\path[draw=drawColor,line width= 0.4pt,dash pattern=on 4pt off 4pt ,line join=round,line cap=round] ( 67.53,138.50) -- ( 67.53,138.50);

\path[draw=drawColor,line width= 0.4pt,line join=round,line cap=round] ( 64.14,108.60) -- ( 70.92,108.60);

\path[draw=drawColor,line width= 0.4pt,line join=round,line cap=round] ( 64.14,138.50) -- ( 70.92,138.50);

\path[draw=drawColor,line width= 0.4pt,line join=round,line cap=round] ( 60.74,119.28) --
	( 74.31,119.28) --
	( 74.31,138.50) --
	( 60.74,138.50) --
	( 60.74,119.28);

\path[draw=drawColor,line width= 1.2pt,line join=round] (128.60,138.50) -- (142.17,138.50);

\path[draw=drawColor,line width= 0.4pt,dash pattern=on 4pt off 4pt ,line join=round,line cap=round] (135.38, 40.24) -- (135.38, 63.74);

\path[draw=drawColor,line width= 0.4pt,dash pattern=on 4pt off 4pt ,line join=round,line cap=round] (135.38,138.50) -- (135.38,138.50);

\path[draw=drawColor,line width= 0.4pt,line join=round,line cap=round] (131.99, 40.24) -- (138.78, 40.24);

\path[draw=drawColor,line width= 0.4pt,line join=round,line cap=round] (131.99,138.50) -- (138.78,138.50);

\path[draw=drawColor,line width= 0.4pt,line join=round,line cap=round] (128.60, 63.74) --
	(142.17, 63.74) --
	(142.17,138.50) --
	(128.60,138.50) --
	(128.60, 63.74);

\path[draw=drawColor,line width= 1.2pt,line join=round] (196.45, 76.55) -- (210.02, 76.55);

\path[draw=drawColor,line width= 0.4pt,dash pattern=on 4pt off 4pt ,line join=round,line cap=round] (203.24, 29.56) -- (203.24, 35.97);

\path[draw=drawColor,line width= 0.4pt,dash pattern=on 4pt off 4pt ,line join=round,line cap=round] (203.24,138.50) -- (203.24,136.37);

\path[draw=drawColor,line width= 0.4pt,line join=round,line cap=round] (199.85, 29.56) -- (206.63, 29.56);

\path[draw=drawColor,line width= 0.4pt,line join=round,line cap=round] (199.85,138.50) -- (206.63,138.50);

\path[draw=drawColor,line width= 0.4pt,line join=round,line cap=round] (196.45, 35.97) --
	(210.02, 35.97) --
	(210.02,136.37) --
	(196.45,136.37) --
	(196.45, 35.97);
\end{scope}
\begin{scope}
\path[clip] (  0.00,  0.00) rectangle (216.81,144.54);
\definecolor{drawColor}{RGB}{0,0,0}

\path[draw=drawColor,line width= 0.4pt,line join=round,line cap=round] ( 33.60, 31.69) -- ( 33.60,138.50);

\path[draw=drawColor,line width= 0.4pt,line join=round,line cap=round] ( 33.60, 31.69) -- ( 29.40, 31.69);

\path[draw=drawColor,line width= 0.4pt,line join=round,line cap=round] ( 33.60, 53.06) -- ( 29.40, 53.06);

\path[draw=drawColor,line width= 0.4pt,line join=round,line cap=round] ( 33.60, 74.42) -- ( 29.40, 74.42);

\path[draw=drawColor,line width= 0.4pt,line join=round,line cap=round] ( 33.60, 95.78) -- ( 29.40, 95.78);

\path[draw=drawColor,line width= 0.4pt,line join=round,line cap=round] ( 33.60,117.14) -- ( 29.40,117.14);

\path[draw=drawColor,line width= 0.4pt,line join=round,line cap=round] ( 33.60,138.50) -- ( 29.40,138.50);

\node[text=drawColor,rotate= 90.00,anchor=base,inner sep=0pt, outer sep=0pt, scale=  1.00] at ( 23.52, 31.69) {0.5};

\node[text=drawColor,rotate= 90.00,anchor=base,inner sep=0pt, outer sep=0pt, scale=  1.00] at ( 23.52, 53.06) {0.6};

\node[text=drawColor,rotate= 90.00,anchor=base,inner sep=0pt, outer sep=0pt, scale=  1.00] at ( 23.52, 74.42) {0.7};

\node[text=drawColor,rotate= 90.00,anchor=base,inner sep=0pt, outer sep=0pt, scale=  1.00] at ( 23.52, 95.78) {0.8};

\node[text=drawColor,rotate= 90.00,anchor=base,inner sep=0pt, outer sep=0pt, scale=  1.00] at ( 23.52,117.14) {0.9};

\node[text=drawColor,rotate= 90.00,anchor=base,inner sep=0pt, outer sep=0pt, scale=  1.00] at ( 23.52,138.50) {1.0};
\end{scope}
\begin{scope}
\path[clip] (  0.00,  0.00) rectangle (216.81,144.54);
\definecolor{drawColor}{RGB}{0,0,0}

\node[text=drawColor,rotate= 90.00,anchor=base,inner sep=0pt, outer sep=0pt, scale=  1.00] at ( 10.92, 84.03) {$|V'|/|V(R)|$};
\end{scope}
\begin{scope}
\path[clip] (  0.00,  0.00) rectangle (216.81,144.54);
\definecolor{drawColor}{RGB}{0,0,0}

\path[draw=drawColor,line width= 0.4pt,line join=round,line cap=round] ( 33.60, 25.20) --
	(216.81, 25.20) --
	(216.81,142.86) --
	( 33.60,142.86) --
	( 33.60, 25.20);
\end{scope}
\begin{scope}
\path[clip] ( 33.60, 25.20) rectangle (216.81,142.86);
\definecolor{drawColor}{RGB}{190,190,190}

\path[draw=drawColor,line width= 0.4pt,dash pattern=on 1pt off 3pt ,line join=round,line cap=round] ( 91.28, 25.20) -- ( 91.28,142.86);

\path[draw=drawColor,line width= 0.4pt,dash pattern=on 1pt off 3pt ,line join=round,line cap=round] (159.13, 25.20) -- (159.13,142.86);

\path[draw=drawColor,line width= 0.4pt,dash pattern=on 1pt off 3pt ,line join=round,line cap=round] ( 33.60, 31.69) -- (216.81, 31.69);

\path[draw=drawColor,line width= 0.4pt,dash pattern=on 1pt off 3pt ,line join=round,line cap=round] ( 33.60, 53.06) -- (216.81, 53.06);

\path[draw=drawColor,line width= 0.4pt,dash pattern=on 1pt off 3pt ,line join=round,line cap=round] ( 33.60, 74.42) -- (216.81, 74.42);

\path[draw=drawColor,line width= 0.4pt,dash pattern=on 1pt off 3pt ,line join=round,line cap=round] ( 33.60, 95.78) -- (216.81, 95.78);

\path[draw=drawColor,line width= 0.4pt,dash pattern=on 1pt off 3pt ,line join=round,line cap=round] ( 33.60,117.14) -- (216.81,117.14);
\definecolor{fillColor}{RGB}{0,0,0}

\path[fill=fillColor] ( 50.07,138.00) rectangle ( 51.07,139.00);

\path[fill=fillColor] ( 44.31,138.00) rectangle ( 45.31,139.00);

\path[fill=fillColor] ( 41.33,138.00) rectangle ( 42.33,139.00);

\path[fill=fillColor] ( 45.00,138.00) rectangle ( 46.00,139.00);

\path[fill=fillColor] ( 49.01,138.00) rectangle ( 50.01,139.00);

\path[fill=fillColor] ( 45.31,138.00) rectangle ( 46.31,139.00);

\path[fill=fillColor] ( 52.32,138.00) rectangle ( 53.32,139.00);

\path[fill=fillColor] ( 51.44,135.87) rectangle ( 52.44,136.87);

\path[fill=fillColor] ( 41.08,131.59) rectangle ( 42.08,132.59);

\path[fill=fillColor] ( 50.54,138.00) rectangle ( 51.54,139.00);

\path[fill=fillColor] ( 50.90,108.10) rectangle ( 51.90,109.10);

\path[fill=fillColor] ( 51.05, 39.74) rectangle ( 52.05, 40.74);

\path[fill=fillColor] ( 41.57,127.32) rectangle ( 42.57,128.32);

\path[fill=fillColor] ( 51.19, 84.60) rectangle ( 52.19, 85.60);

\path[fill=fillColor] ( 43.49,110.23) rectangle ( 44.49,111.23);

\path[fill=fillColor] (118.94,138.00) rectangle (119.94,139.00);

\path[fill=fillColor] (108.38,138.00) rectangle (109.38,139.00);

\path[fill=fillColor] (112.91,138.00) rectangle (113.91,139.00);

\path[fill=fillColor] (112.90,138.00) rectangle (113.90,139.00);

\path[fill=fillColor] (109.29,138.00) rectangle (110.29,139.00);

\path[fill=fillColor] (111.57,138.00) rectangle (112.57,139.00);

\path[fill=fillColor] (111.58,138.00) rectangle (112.58,139.00);

\path[fill=fillColor] (119.04,138.00) rectangle (120.04,139.00);

\path[fill=fillColor] (121.18,138.00) rectangle (122.18,139.00);

\path[fill=fillColor] (115.15,138.00) rectangle (116.15,139.00);

\path[fill=fillColor] (119.06, 65.37) rectangle (120.06, 66.37);

\path[fill=fillColor] (111.85, 61.10) rectangle (112.85, 62.10);

\path[fill=fillColor] (110.27, 56.83) rectangle (111.27, 57.83);

\path[fill=fillColor] (120.22, 39.74) rectangle (121.22, 40.74);

\path[fill=fillColor] (111.27, 48.28) rectangle (112.27, 49.28);

\path[fill=fillColor] (188.49,138.00) rectangle (189.49,139.00);

\path[fill=fillColor] (179.40,133.73) rectangle (180.40,134.73);

\path[fill=fillColor] (183.05, 71.78) rectangle (184.05, 72.78);

\path[fill=fillColor] (188.59,138.00) rectangle (189.59,139.00);

\path[fill=fillColor] (183.43, 80.33) rectangle (184.43, 81.33);

\path[fill=fillColor] (180.82, 29.06) rectangle (181.82, 30.06);

\path[fill=fillColor] (183.66,138.00) rectangle (184.66,139.00);

\path[fill=fillColor] (177.39, 39.74) rectangle (178.39, 40.74);

\path[fill=fillColor] (181.62, 33.33) rectangle (182.62, 34.33);

\path[fill=fillColor] (184.68,133.73) rectangle (185.68,134.73);

\path[fill=fillColor] (187.17, 35.47) rectangle (188.17, 36.47);

\path[fill=fillColor] (179.39, 31.19) rectangle (180.39, 32.19);

\path[fill=fillColor] (186.23,138.00) rectangle (187.23,139.00);

\path[fill=fillColor] (180.27,133.73) rectangle (181.27,134.73);

\path[fill=fillColor] (182.22, 46.15) rectangle (183.22, 47.15);

\path[fill=fillColor] (180.25,138.00) rectangle (181.25,139.00);

\path[fill=fillColor] (187.86,120.91) rectangle (188.86,121.91);

\path[fill=fillColor] (176.03, 29.06) rectangle (177.03, 30.06);

\path[fill=fillColor] (179.21,135.87) rectangle (180.21,136.87);

\path[fill=fillColor] (183.82, 48.28) rectangle (184.82, 49.28);

\path[fill=fillColor] (181.05, 35.47) rectangle (182.05, 36.47);

\path[fill=fillColor] (175.78,138.00) rectangle (176.78,139.00);

\path[fill=fillColor] (185.95, 35.47) rectangle (186.95, 36.47);

\path[fill=fillColor] (184.68, 37.60) rectangle (185.68, 38.60);

\path[fill=fillColor] (180.55,138.00) rectangle (181.55,139.00);

\path[fill=fillColor] (179.03,131.59) rectangle (180.03,132.59);

\path[fill=fillColor] (184.23, 52.56) rectangle (185.23, 53.56);

\path[fill=fillColor] (180.52,135.87) rectangle (181.52,136.87);

\path[fill=fillColor] (182.05,120.91) rectangle (183.05,121.91);

\path[fill=fillColor] (180.79, 35.47) rectangle (181.79, 36.47);

\path[fill=fillColor] (186.60,138.00) rectangle (187.60,139.00);

\path[fill=fillColor] (187.30, 39.74) rectangle (188.30, 40.74);

\path[fill=fillColor] (179.17, 35.47) rectangle (180.17, 36.47);

\path[fill=fillColor] (178.76,135.87) rectangle (179.76,136.87);

\path[fill=fillColor] (177.37, 37.60) rectangle (178.37, 38.60);

\path[fill=fillColor] (177.89, 31.19) rectangle (178.89, 32.19);
\end{scope}
\begin{scope}
\path[clip] (  0.00,  0.00) rectangle (216.81,144.54);
\definecolor{drawColor}{RGB}{0,0,0}

\path[draw=drawColor,line width= 0.4pt,line join=round,line cap=round] ( 57.35, 25.20) -- (193.06, 25.20);

\path[draw=drawColor,line width= 0.4pt,line join=round,line cap=round] ( 57.35, 25.20) -- ( 57.35, 21.00);

\path[draw=drawColor,line width= 0.4pt,line join=round,line cap=round] (125.20, 25.20) -- (125.20, 21.00);

\path[draw=drawColor,line width= 0.4pt,line join=round,line cap=round] (193.06, 25.20) -- (193.06, 21.00);

\node[text=drawColor,anchor=base,inner sep=0pt, outer sep=0pt, scale=  1.00] at ( 57.35, 10.08) {alba-$p$-$i$};

\node[text=drawColor,anchor=base,inner sep=0pt, outer sep=0pt, scale=  1.00] at (125.20, 10.08) {madr-$p$-$i$};

\node[text=drawColor,anchor=base,inner sep=0pt, outer sep=0pt, scale=  1.00] at (193.06, 10.08) {ur-$n$-$d$-$p$};
\end{scope}
\end{tikzpicture}
   \caption{Data reduction effect of our PSAKS
    relative to the total number of input vertices (left)
    and to the number of vertices incident to required edges (right).
    Each dot represents an instance.
    The boxes show the first quartile, the median, and the third quartile.
    The whiskers extend up to 1.5 times the interquartile range.}
  \label{fig:quartiles}
\end{figure*}
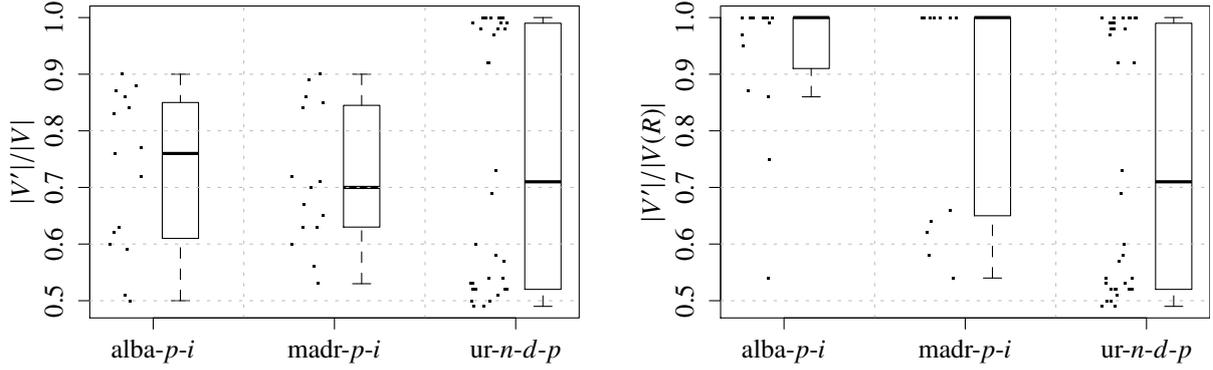

\paragraph{Experimental results.}
\cref{fig:quartiles} gives a rough idea
of the data reduction effect of our PSAKS
on the alba-$p$-$i$, madr-$p$-$i$, and ur-$n$-$d$-$p$ instances.
The complete experimental results,
including the two Berlin instances,
are shown in \cref{tab:alba,tab:real,tab:berlin,tab:random},
where additionally to the notation in \cref{def:bounds},
we denote by
\begin{compactitem}[$|V'|, |R'|$]
\item[$\cost(W)$] -- the weight of a 3/2-approximation computed in the input graph,
\item[$\cost(W')$] -- the weight of a 3/2-approximation computed in the kernel and lifted to the input graph,
\item[$|V'|, |R'|$] -- the number of vertices and required edges in the kernel, respectively, and by
\item[ms] -- the number of milliseconds it took to compute the kernel (not counting the time for computing pairwise shortest path lengths for
  establishing the triangle inequality using \cref{lem:triangle}).
\end{compactitem}
Since for the ur-$n$-$d$-$p$ instances,
the weight of an optimal solution is known,
we compare $\cost(W')$ to the optimum
in \cref{tab:random}.\footnote{For the alba-$p$-$i$ and madr-$p$-$i$ instances,
  the optimum is only known when considered as instances of the General Routing Problem,
  where also vertices have to be visited.}
Remarkably,
the best compression results are achieved on the Berlin instances,
the only instances consisting purely out of real\hyp world data: only 22\,\% of
the vertices incident to required edges remain.
\cref{fig:berlin} visualizes the kernelization effect on two
strongly
compressed instances.

In the following,
we discuss the data reduction effect in more detail
with the help
of the plots in \cref{fig:effect}.
Since some of the instances in the literature
(concretely, the ur-$n$-$d$-$p$ instances)
are already preprocessed with respect
to \cref{rrule:delverts},
we analyze the data reduction effect
with respect to the number~$|V(R)|$ of those input
vertices incident to required edges.
For the sake of completeness,
the data reduction effect
with respect to the number of \emph{all} input vertices
is shown in \cref{tab:alba,tab:real,tab:berlin,tab:random}.

\looseness=-1
In \cref{fig:to-v,fig:to-r},
we see that the effectivity of our PSAKS grows
with the number of input vertices incident to required edges
and with the number of required edges themselves.
In \cref{fig:to-c},
we see that, as expected,
the effectivity decreases as the number of connected components of~$G\langle R\rangle$ grows.
However,
in all three plots,
we see a clear clustering
of different instance types. It is \cref{fig:to-cavg}
that shows the determining feature for
the effectivity of our PSAKS:
in all instances types,
it
uniformly grows with the
average size of the connected components of~$G\langle R\rangle$.
This comes at no surprise,
since the main action of our PSAKS
is shrinking these connected components.

Regarding the solution quality, we point out that,
despite kernelizing all instances with $\varepsilon=1/10$
and thus allowing a weight increase by a factor of $1.1$
when lifting a solution from the kernel to the original instance,
the maximal such weight increase observed is 1.01 (for the alba-$p$-$i$ instances in \cref{tab:alba}),
whereas often no weight increase is observed.
In some cases,
kernelization leads to better solutions
(for the instances in \cref{tab:real,tab:random}).
Also, in \cref{tab:random},
the 3/2-approximate solution lifted from the kernel
turns out to be worse than the optimum by a factor not larger than 1.075
and is thus way below the allowance.

\begin{figure}[p]
  \begin{subfigure}[b]{0.45\textwidth}
    \includegraphics[width=\textwidth]{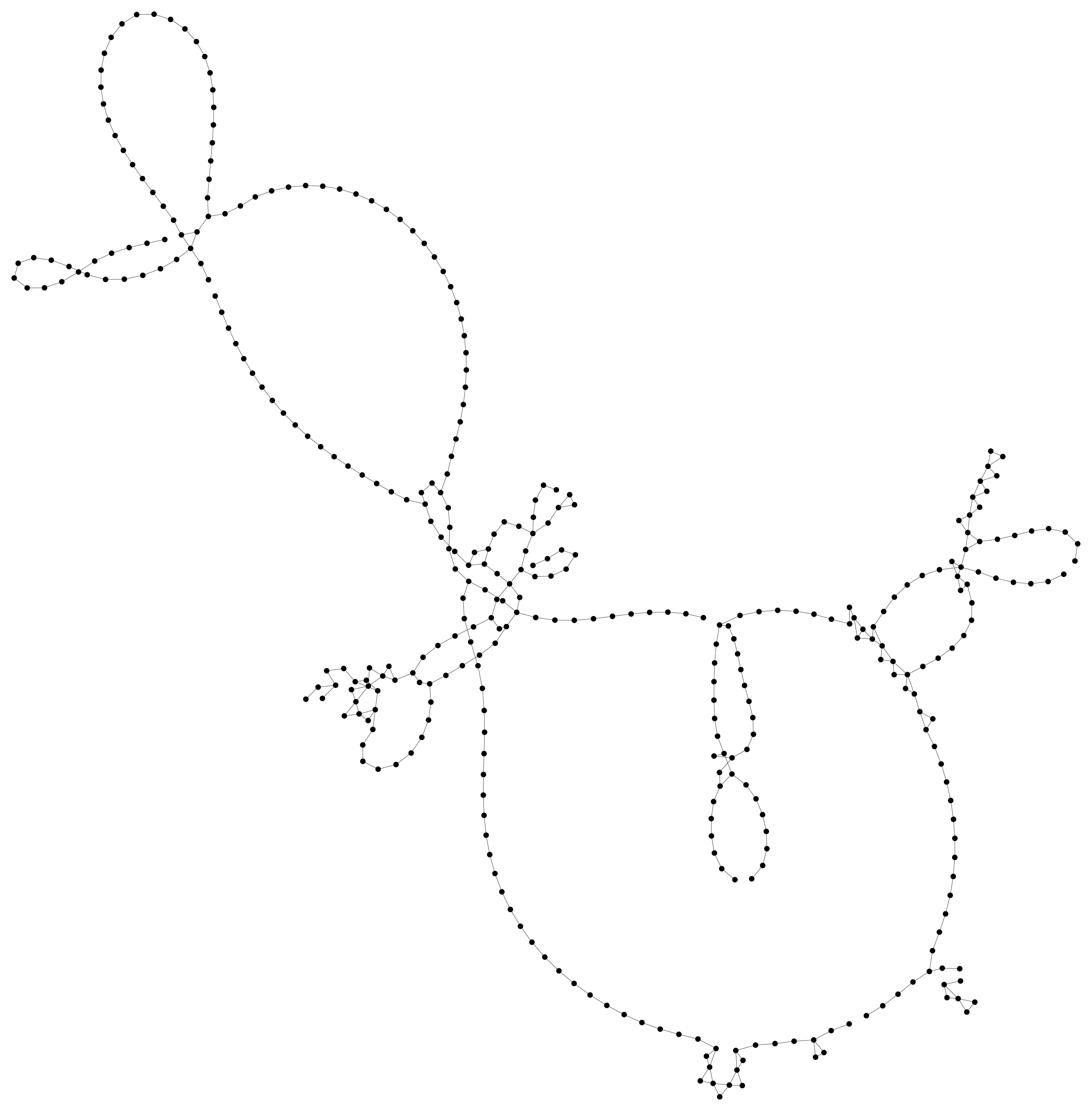}
    \caption{Berlin instance before kernelization.}
  \end{subfigure}\hfill
  \begin{subfigure}[b]{0.45\textwidth}
    \includegraphics[width=\textwidth]{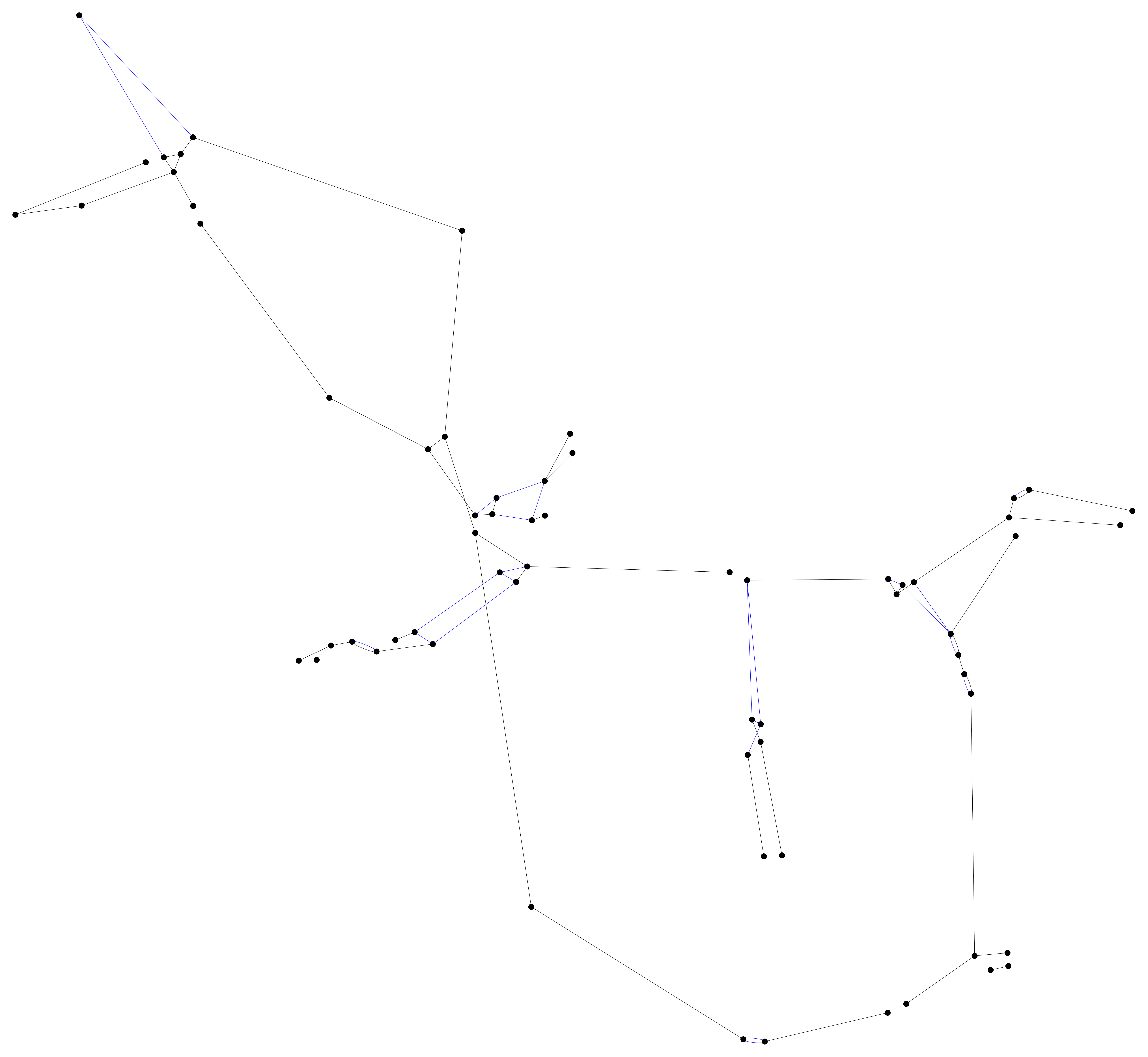}
    \caption{Berlin instance after kernelization.}
  \end{subfigure}
  \\[2cm]
  \begin{subfigure}[b]{0.45\textwidth}
    \includegraphics[width=\textwidth]{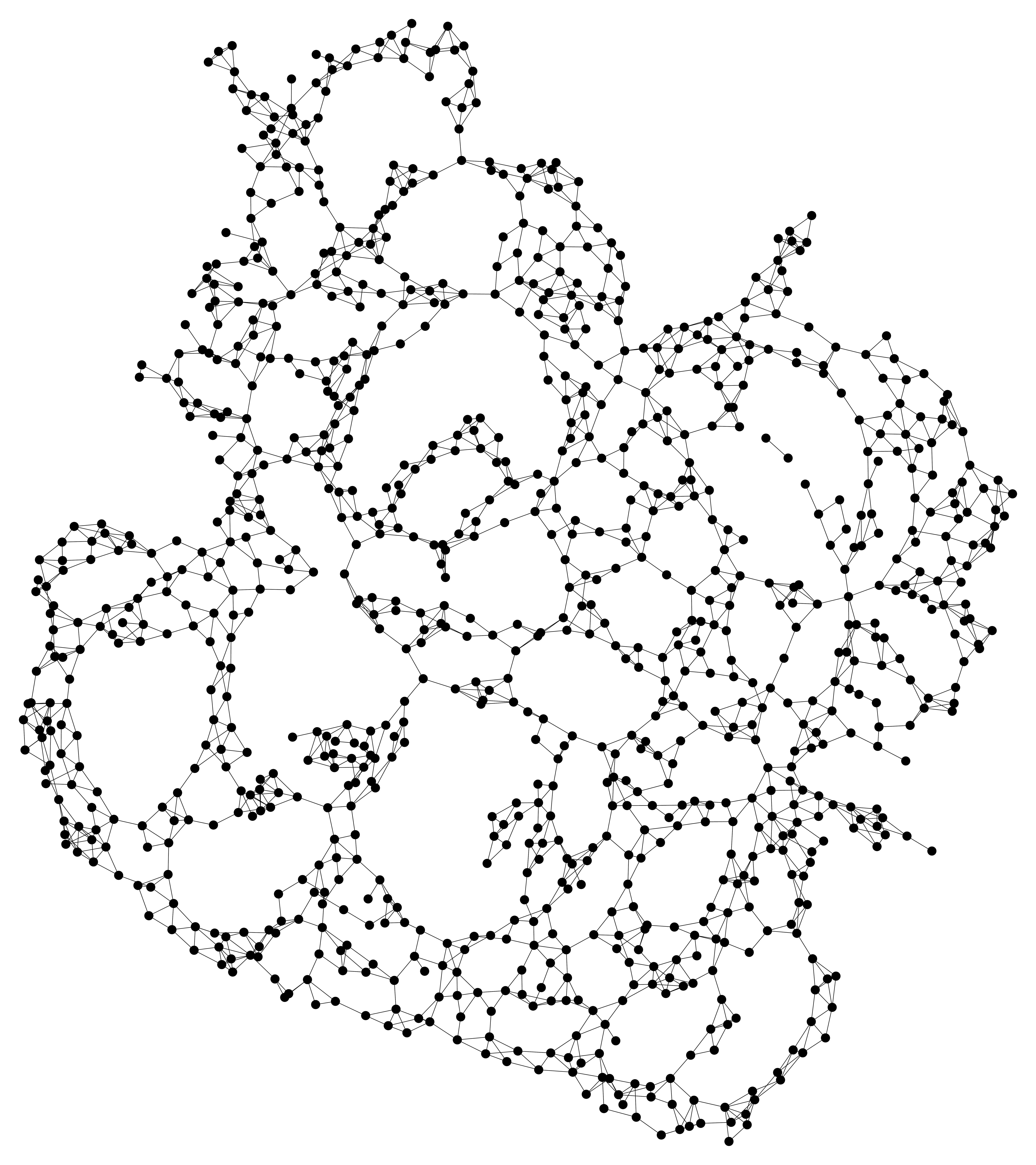}
    \caption{The ur-$1000$-$5$-$0.75$ instance before kernelization.}
  \end{subfigure}\hfill
  \begin{subfigure}[b]{0.45\textwidth}
    \includegraphics[width=\textwidth]{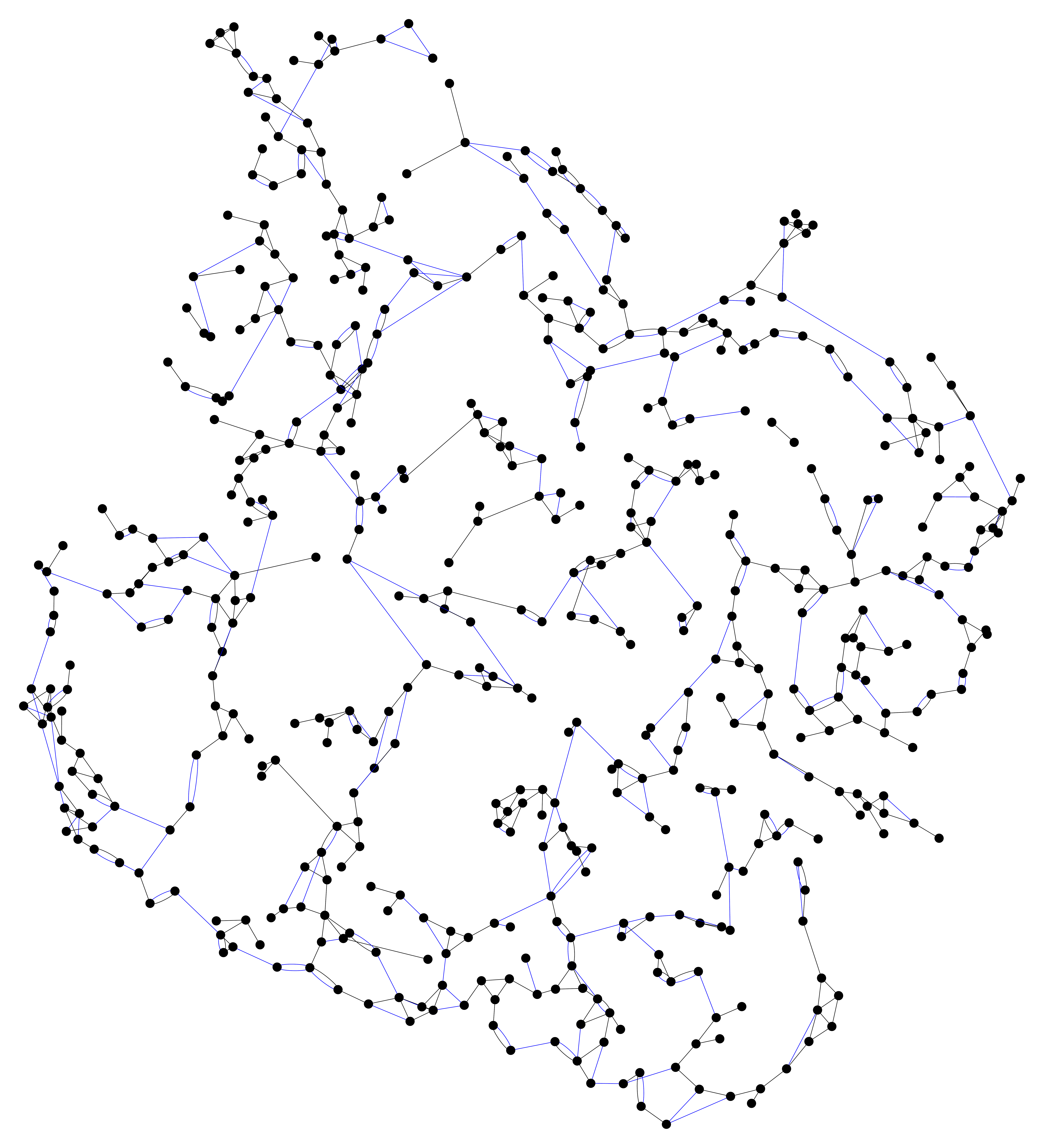}
    \caption{The ur-$1000$-$5$-$0.75$ instance after kernelization.}
  \end{subfigure}
  \caption{Two instances before (left) and after (right) kernelization.
    Only required edges are shown.
    Blue edges are the matching edges added by \cref{rrule:ballred}.}
  \label{fig:berlin}
\end{figure}

\begin{figure}[p]
  \centering
  \ref{mylegend}
  \\[1cm]
  \begin{subfigure}[b]{0.45\textwidth}
    \begin{tikzpicture}[baseline]
      \begin{axis}[ylabel={$|V'|/|V(R)|$}, xmin=-1,
        xlabel={$|V(R)|$},
        ylabel near ticks,
        legend columns=-1,
        legend entries={alba-$p$-$i$\vspace{1cm}, madr-$p$-$i$\quad, ur-500-$d$-$p$\quad, ur-700-$d$-$p$\quad, ur-1000-$d$-$p$},
        legend to name=mylegend,
        legend style={/tikz/every even column/.append style={column sep=0.5cm}},
        legend style={/tikz/every odd column/.append style={column sep=0.15cm}}]
        \addplot[only marks , mark=o]       table[col sep=ampersand, x={nr}, y={n2nr}] {alba.texinclude};
        \addplot[only marks , mark=triangle]       table[col sep=ampersand, x={nr}, y={n2nr}] {madr.texinclude};
        \addplot[only marks , mark=*]       table[col sep=ampersand, x={nr}, y={n2nr}] {ur500.texinclude};
        \addplot[only marks , mark=star]    table[col sep=ampersand, x={nr}, y={n2nr}] {ur700.texinclude};
        \addplot[only marks , mark=+]       table[col sep=ampersand, x={nr}, y={n2nr}] {ur1000.texinclude};
\end{axis}
    \end{tikzpicture}
    \caption{Number of input vertices incident to required edges vs fraction of remaining vertices.}
    \label{fig:to-v}
  \end{subfigure}
  \hfill
  \begin{subfigure}[b]{0.45\textwidth}
    \begin{tikzpicture}[baseline]
      \begin{axis}[ylabel={$|R'|/|R|$},
        xlabel={$|R|$},
        xmax=2100,xmin=-1,
        yticklabel pos=upper,
        ylabel near ticks]
        \addplot[only marks , mark=o]       table[col sep=ampersand, x={r1}, y={r2r1}] {alba.texinclude};
        \addplot[only marks , mark=triangle]       table[col sep=ampersand, x={r1}, y={r2r1}] {madr.texinclude};
        \addplot[only marks , mark=*]       table[col sep=ampersand, x={r1}, y={r2r1}] {ur500.texinclude};
        \addplot[only marks , mark=star]    table[col sep=ampersand, x={r1}, y={r2r1}] {ur700.texinclude};
        \addplot[only marks , mark=+]       table[col sep=ampersand, x={r1}, y={r2r1}] {ur1000.texinclude};

\end{axis}
    \end{tikzpicture}
\caption{Number of input required edges vs fraction of remaining required edges.}
    \label{fig:to-r}
  \end{subfigure}
  \\[1cm]
  \begin{subfigure}[b]{0.45\textwidth}
    \begin{tikzpicture}[baseline]
      \begin{axis}[xmax=90, xmin=-1, ylabel={$|V'|/|V(R)|$},
        xlabel={$\phantom{|V|}c$},
        ylabel near ticks]
        \addplot[only marks , mark=o]       table[col sep=ampersand, x={c}, y={n2nr}] {alba.texinclude};
        \addplot[only marks , mark=triangle]       table[col sep=ampersand, x={c}, y={n2nr}] {madr.texinclude};

        \addplot[only marks , mark=*]       table[col sep=ampersand, x={c}, y={n2nr}] {ur500.texinclude};
        \addplot[only marks , mark=star]    table[col sep=ampersand, x={c}, y={n2nr}] {ur700.texinclude};
        \addplot[only marks , mark=+]       table[col sep=ampersand, x={c}, y={n2nr}] {ur1000.texinclude};

\end{axis}
    \end{tikzpicture}
    \caption{Number of connected components of~$G\langle R\rangle$ vs fraction of remaining vertices.}
    \label{fig:to-c}
  \end{subfigure}
\hfill
\begin{subfigure}[b]{0.45\textwidth}
\begin{tikzpicture}[baseline]
  \begin{axis}[ylabel={$|V'|/|V(R)|$},
    xlabel={$|V(R)|/c$},
    xmax=260, xmin=0,        yticklabel pos=upper,
        ylabel near ticks]
    \addplot[only marks , mark=o]       table[col sep=ampersand, x expr=\thisrow{nr}/\thisrow{c}, y={n2nr}] {alba.texinclude};
    \addplot[only marks , mark=triangle]       table[col sep=ampersand, x expr=\thisrow{nr}/\thisrow{c}, y={n2nr}] {madr.texinclude};
    \addplot[only marks , mark=*]       table[col sep=ampersand, x expr=\thisrow{nr}/\thisrow{c}, y={n2nr}] {ur500.texinclude};
    \addplot[only marks , mark=star]    table[col sep=ampersand, x expr=\thisrow{nr}/\thisrow{c}, y={n2nr}] {ur700.texinclude};
    \addplot[only marks , mark=+]       table[col sep=ampersand, x expr=\thisrow{nr}/\thisrow{c}, y={n2nr}] {ur1000.texinclude};

\end{axis}
\end{tikzpicture}
\caption{Average size of connected component in $G\langle R\rangle$ vs fraction of remaining vertices.}
\label{fig:to-cavg}
\end{subfigure}

\caption{Effect of data reduction of our PSAKS.}
\label{fig:effect}
\end{figure}

\paragraph{Possible improvements.}
The effectivity of our data reduction can be increased
replacing $\cost(R)$ by
\[
  \max\Biggl\{\cost(R)+\cost(M),\cost(R)+\cost(T),\cost(R)+\frac{\cost(M)+\cost(T)}2\Biggr\}
\]
in the choice of $\gamma$ in~\eqref{gamma} for the application of \cref{rrule:ballred}.
Since this also is a lower bound for $\cost(W^*)$ (recall \cref{lem:bounds}),
such a replacement will still guarantee
that a $\alpha(1+\varepsilon)$-approximation
can be lifted from a $\alpha$-approximation on the kernel.
However,
this replacement removes about one or two percents of vertices more,
whereas
computing $\cost(M)$ in the larger instances
took between 11 and 40 seconds,
so the pay\hyp off is very limited.

\section{Conclusion}
Our main algorithmic contribution
is a polynomial\hyp size approximate kernelization scheme (PSAKS)
for the Rural Postman Problem
parameterized by~$b+c$,
where $b$~is the number of vertices
incident to an odd number of required edges
and $c$~is the number of connected components
formed by the required edges.
Experiments show that
the data reduction algorithm
efficiently shrinks
problem instances with few connected components
without largely sacrificing solution quality.
We also showed a PSAKS
for the parameter $\cost(T)$,
which gives a PSAKS for the parameter~$c$
when edge weights are bounded polynomially in~$c$.
These results together naturally lead to the question
whether a PSAKS for the parameter~$c$ exists
(we conjecture ``yes'').

We think that the approach taken by
\cref{rrule:ballred},
namely reducing all vertices that
do not belong to some inclusion\hyp maximal set~$B$
of mutually sufficiently distant vertices,
might be applicable to other metric graph problems:
it ensures that,
for each deleted vertex,
some nearby representative in~$B$ is retained.
In preliminary research, for example,
we also found it to applicable
to a metric variant
of the \textsc{Min-Power Symmetric Connectivity} problem
where it is required to connect $c$~disconnected
parts of a wireless sensor network \citep{BBNN17}
and to the Location Rural Postman Problem \citep{BT19}.
Notably,
this approach does not generalize well to
asymmetric distances,
so that another vexing question
besides proving \cref{mainconj}
is whether the scheme for the parameter~$b+c$
presented in this work can be generalized
to the \emph{directed} Rural Postman Problem.
We point out that,
using known ideas \citep{BKS17},
one can reduce any instance~$I$ of the directed
or undirected RPP to an instance~$I'$
with $c$~vertices in $O(n^3\log n)$~time
such that any \(\alpha\)-approximation for~$I'$
yields an \((\alpha+1)\)-approximation for~$I$.
Given that undirected RPP is 3/2-approximable,
this is interesting only for the directed~RPP.

\paragraph{Acknowledgments.}
We thank the anonymous referees of \emph{Networks} for their constructive feedback.

\paragraph{Funding.}
\looseness=-1
R.\ van Bevern and O.\ Yu.\ Tsidulko
are supported by the Russian Foundation for Basic Research,
project~18-501-12031 NNIO\textunderscore a.
T.\ Fluschnik
is supported by the German Research Foundation,
project TORE (NI~369/18).

\small
\bibliographystyle{dr-rpp}

\end{document}